\documentclass[reprint,superscriptaddress,amsmath,amssymb,aps,prl]{revtex4-1}
\usepackage{calc}
\usepackage{textcomp}
\usepackage{amsmath}
\usepackage{amssymb}
\usepackage{graphicx}
\usepackage{natbib}
\usepackage{multibib}
\usepackage{color}

\usepackage{float}
\usepackage{amsmath,amsthm,amsfonts,amssymb,amscd}

\usepackage[a4paper, left=1.5cm, right=1cm, top=2.5cm, bottom=2.5cm, headsep=1.2cm]{geometry}

\usepackage[utf8]{inputenc} 
\usepackage{tabularx}
\usepackage{float}
\usepackage{multirow}
\usepackage{longtable}

\newtheorem*{lemma}{Lemma}

\newtheorem*{prop}{Proposition}
\newtheorem{theorem}{Theorem}
\newtheorem{dfn}{Definition}

\makeatother

\begin{document}
	
	\title{A No-go theorem for device-independent security in relativistic causal theories}
	\author{R. Salazar}
	
	\email{rb.salazar.vargas@gmail.com}
	
	\affiliation{Institute of Theoretical Physics and Astrophysics, National Quantum Information Centre, University of Gdansk, 80-308 Gdansk, Poland}
	
	\affiliation{Faculty of Applied Physics and Mathematics, National Quantum Information Centre, Gdansk University of Technology, 80-233 Gdansk, Poland}
	
	\affiliation{ Institute of Informatics Faculty of Mathematics, Physics and Informatics,
		University of Gdansk, 80-308 Gdansk, Poland}
	
	\affiliation{Institute of Physics, Jagiellonian University, 30-059 Krakow, Poland}

	\author{M. Kamo\'{n}}
	
	\affiliation{Faculty of Applied Physics and Mathematics, National Quantum Information Centre, Gdansk University of Technology, 80-233 Gdansk, Poland}

	\author{K. Horodecki}
	\affiliation{ Institute of Informatics Faculty of Mathematics, Physics and Informatics,
		University of Gdansk, 80-308 Gdansk, Poland}
	\affiliation{International Centre for Theory of Quantum Technologies,
		University of Gdansk, Wita Stwosza 63, 80-308 Gdansk, Poland}
	
	\author{D. Goyeneche}
	\affiliation{Institute of Physics, Jagiellonian University, 30-059 Krakow, Poland}
	\affiliation{Departamento de F\'{i}sica, Facultad de Ciencias B\'{a}sicas, Universidad de Antofagasta, Casilla 170, Antofagasta, Chile}
	
	\author{D. Saha}
	\affiliation{Institute of Theoretical Physics and Astrophysics, National Quantum Information Centre, University of Gdansk, 80-308 Gdansk, Poland}
          \affiliation{Center for Theoretical Physics, Polish Academy of Sciences, Aleja Lotnik\'{o}w 32/46, 02-668 Warsaw, Poland}
	
	\author{R. Ramanathan}
	\affiliation{Laboratoire d’Information Quantique, Universite Libre de Bruxelles, Belgium}

	\author{P. Horodecki}
	\affiliation{Faculty of Applied Physics and Mathematics, National Quantum Information Centre, Gdansk University of Technology, 80-233 Gdansk, Poland}
	\affiliation{International Centre for Theory of Quantum Technologies,
		University of Gdansk, Wita Stwosza 63, 80-308 Gdansk, Poland}
	
	\begin{abstract}
		
		
		A  crucial task for secure communication networks is to determine the minimum of physical requirements to certify a cryptographic protocol. A widely accepted candidate for certification is the principle of relativistic causality which is equivalent to the disallowance of causal loops. Contrary to expectations, we demonstrate how correlations allowed by relativistic causality could be exploited to break security for a broad class of multi-party protocols (all modern protocols belong to this class).  As we show, deep roots of this dramatic lack of security lies in the fact that unlike in previous (quantum or no-signaling) scenarios the new theory ,, decouples" the property of extremality and that of statistical independence on environment variables. Finally, we find out, that the lack of security is accompanied by some advantage: the new correlations can reduce communication complexity better than the no-signaling ones. As a tool for analysis of this advantage, we characterise relativistic causal polytope by its extremal points in the simplest multi-party scenario that goes beyond the no-signaling paradigm.
		
	\end{abstract}
	
	\date{\today}
	\maketitle
	
	\paragraph{Introduction}
	
	Cryptography covers a plethora of security scenarios, ranging from secure key distribution via protocols such as secret sharing to the two-party cryptographic protocols. These include, e.g., bit commitment or anonymous voting.  However, no doubt the greatest impact on physics that it has, is due to the foundational role of Cryptography in the development of the field of Quantum Information (QI).
	Due to the seminal ideas of Wiesner \cite{Wiesner} and later Bennett and Brassard \cite{BB84}, Quantum Cryptography became a pillar of QI, which upgraded security based on computer assumptions to the one founded on physical laws - that of Quantum Mechanics (QM). 
	
	On the other hand, in parallel, the evolution of QI has led to the relaxation of properties of QM that opens the possibility for new theories (NT) beyond QM, such as e.g., Generalized Probabilistic Theories \cite{Bell-nonlocality}. In this direction, it is compelling to ask:  Does the new theory allow for secure protocols? And if so, in which scenarios? In particular, there is an important question: {\it is, in a given NT, the {\it device-independent} (DI) framework for security certification? (e.g., for the recent development of DI framework  within Quantum Mechanics see \cite{RotemDupuisFawziRenner} and references therein)}. The first research in this direction was the development of a scenario within GPT that includes the so-called {\it non-signaling adversary} - leading to the so-called non-signaling device-independent security (NSDI) \cite{Kent, Kent-Colbeck, Scarani2006, acin-2006-8,AcinGM-bellqkd, masanes-2009-102, hanggi-2009,lit13}. The NSDI scenario relaxes the requirement that the adversary's knowledge and technical skills are bound to quantum theory. The adversary has access to additional resources that are only limited by the physical principle of {\it no-faster than light communication}.
	It is known that, e.g., secure key distribution can be achieved in the NSDI scenario if the devices at the hand of the honest parties are measured at once in parallel \cite{acin-2006-8,AcinGM-bellqkd,masanes-2009-102, hanggi-2009,lit13}. On the other hand, it is believed that attacks based on forward-signaling between the rounds of the experiment can be fatal \cite{Rotem12,Salwey-Wolf}. 
	In what follows, we investigate a recently proposed framework of  Relativistic Causal (RC) theories that goes beyond the no-signaling paradigm.
	We prove some remarkable properties of these theories that result in severe and possibly fatal limitations on the security if we lack additional information about details of the particular RC theory governing physical reality.
	
	Natural law for an NT is to impose the mentioned axiom, that the speed of light $c$ is a limit for
	the speed of communication between two distant parties to  guarantee lack of logical paradoxes like the famous
	{\it grandfather's   paradox} \cite{Grandfather}. Nevertheless,  it was noticed that resources
	with  $c$ as bound for communication's velocity, can {\it in principle} influence
	correlations in a {\it faster than light} manner if located in special space-time
	configurations \cite{Grunhaus,ref1}. The novel scenario 
	admitting those effects - called {\it relativistic causal} as it prevents 
	any causal paradoxes - implies  unexpected  correlation behaviors \cite{ref1}.
	The crucial element here is that if one party  manipulates in faster-than-light
	manner correlations shared by the other parties, the latter can notice it only after mutual communication (limited by the speed of light) or when they meet together.      Here we show that all the cryptographic protocols known to date fail.
	
	\paragraph{Main results} 
	
	Here we prove a fundamental security no-go theorem against two cooperative adversaries who are constrained by relativistic causality
	: they can break secure key distribution by designing devices with 
	correlations as strong as allowed by relativistic causality.
	Our result determines a significant limitation for the security of a broad class of secure key distribution device-independent protocols certified by relativistic causality alone.
	
	We first concentrate on the phenomenon of monogamy of correlations that underpins the security of DI protocols. We present spatio-temporal configurations of measurement events for which monogamy relationships are broken in the relativistic causality setting. In particular, we show that every two-party Bell inequality becomes completely non-monogamous in these configurations.
	This fact makes a crucial difference between relativistic causal theory and that of no-signaling.  
	In the latter, adequately understood extremality of correlations was equivalent to a complete lack of correlations with any external environment. In relativistic causal theories, very strong correlations can always be present due to the above result.
	Moreover, we establish a hacking strategy for two eavesdroppers who exploit the strongest correlations
	allowed by relativistic causality.     The attack is fatal because the eavesdroppers can learn a copy of the honest parties' correlations as a shared secret that they learn together. The secret sharing structure guarantees no faster-than-light communication.
	
	As we show, this hacking strategy breaks any device-independent security protocol, which begins with the same parallel measurement on a device. These protocols consist of the operations called {\it Measurement on Device followed by Local Operations and Public Communication} (MDLOPC) \cite{NSDI}.
	We note here that all known protocols secure against the non-signaling adversary, perform MDLOPC operations (see e.g. \cite{masanes-2009-102,Renner-Hanggi,hanggi-2009}). Indeed, there are known successful attacks of the no-signaling adversary on protocols with sequential rather than parallel measurement \cite{Rotem-Sha,Salwey-Wolf}.

	We prove the no-go with the help of the link between NSDI and secure key agreement \cite{Maurer93} found in \cite{NSDI}.
	We give detailed proof for the case of two honest parties and show how to extend it to the multipartite case (of the so-colled conference key agreement).

	Finally, as a step towards determining the full potential of a single RC eavesdropper, we complete our analysis with a full
	characterization of the simplest relevant setting under relativistic causal constraints, namely the Bell scenario
	of three parties, each performing two binary measurements.
	We prove its advantage in communication complexity reduction.

	\paragraph{Relativistic Causality vs. No-Signaling}
	
	By the very definition, the no-signaling constraints stipulate that the output distributions of any subset of parties are independent of the choices of the inputs of the remaining parties. 
	
	Relativistic causality is the physical principle which states that \textit{ an effect cannot occur from a cause that is not in its past light cone, and similarly a cause cannot have an effect outside its future light cone, i.e., that there be no causal loops in theory} \cite{ref12, ref13}. The no-signaling constraints given in (1)  are sufficient to ensure that theory respects causality, but as shown in \cite{ref1}, there exist space-time configurations of measurement events for the three parties where not all the above constraints are \textit{necessary} to enforce causality. In particular, consider the measurement configuration of Fig. \ref{Figure1}, where the intersection of the future light cones of Alice (A) and Charlie's (C) measurement events is contained within the future light cone of Bob's (B) measurement event.  Here no breakdown of causality occurs if the joint distribution of the outcomes $a, c$ depended on the input $y$ of Bob, provided that the marginal distributions of $a$ and $c$ separately are independent of $y$. This may be intuitively understood from the fact that the information concerning the correlations between $a$ and $c$ is only accessible at a point in the intersection of the future light cones of Alice and Charlie's measurement events. This intersection in this configuration is contained within the future light cone of Bob's measurement event (for a full proof see \cite{ref1, ref14}).
	In other words, the constraints that are both necessary and sufficient for relativistic causality here are 
	
	\begin{eqnarray}
	\label{eq:RC-constraints}
	\sum_{a}P\left(a,b,c\mid x,y,z\right) & = & \sum_{a}P\left(a,b,c\mid x^{\prime},y,z\right)\:\forall \text{\footnotesize $x,x^{\prime},y,z,b,c$}\nonumber \\
	\sum_{c}P\left(a,b,c\mid x,y,z\right) & = & \sum_{c}P\left(a,b,c\mid x,y,z^{\prime}\right)\:\forall \text{\footnotesize $z,z^{\prime},x,y,a,b$}\nonumber \\
	\sum_{b,c}P\left(a,b,c\mid x,y,z\right) & = & \sum_{b,c}P\left(a,b,c\mid x,y^{\prime},z^{\prime}\right)\:\forall \text{\footnotesize $y,y^{\prime},z,z^{\prime},x,a$}\nonumber \\
	\sum_{a,b}P\left(a,b,c\mid x,y,z\right) & = & \sum_{a,b}P\left(a,b,c\mid x^{\prime},y^{\prime},z\right)\:\forall \text{\footnotesize $x,x^{\prime},y,y^{\prime},z,c$} \nonumber \\
	&  & \label{eq:RCcond}
	\end{eqnarray}
	
	Observe that in the above, there is no sum over a pair (a,c). This fact implies a violation of the no-signaling constraints stipulating that the output distributions of {\it any subset of parties} are independent of the choices of the inputs of {\it the remaining parties}. Consequently, the Relativistic Causal polytope of behaviors is richer in structure and of higher dimensionality than the usual no-signaling polytope for this Bell scenario. Notice also that the above constraints are manifestly Lorentz covariant. If the intersection of the future light cones of $A$ and $C$ is contained within $ B$'s future light cone in one inertial reference frame, then this intersection is contained within the future light cone of $B$ in all inertial reference frames. In Appendix A,  we provide the necessary and sufficient constraints imposed by relativistic causality for an arbitrary number of parties in arbitrary globally hyperbolic space-time \cite{ref14}.
	
	\begin{figure}
		\begin{centering}
			\includegraphics[width=0.9\columnwidth]{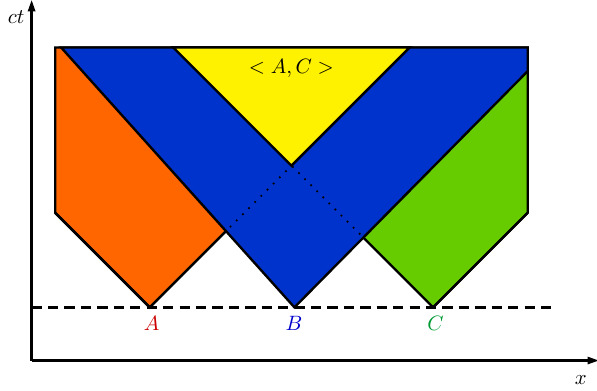}
			\par\end{centering}
		\caption{\label{Figure1} Spacetime configuration of measurement events in the three-party Bell experiment
			when measurements are simultaneous in a common reference frame. The information about two-point Alice and Charlie outputs correlations $ <A,C>$ 
			is only accessible in the yellow region that belongs to causal future of B.}    
	\end{figure}

	\paragraph{ monogamy of non-locality}
	
	One of the most intriguing properties of quantum non-local correlations are monogamy relations. These relations were first observed by Toner for the well-known CHSH inequality \cite{ref27}. These are direct trade-off relations between the amount of violation of an inequality observed by a pair of agents Alice and Bob and the correlations between Alice and Charlie's outcomes. The monogamy of non-locality gives rise to non-trivial bounds on cloning \cite{ref3},  underpins the security of device-independent key distribution and randomness generation protocols against no-signaling adversaries \cite{ref6, ref7} and may help to detect gravitational decoherence \cite{ref33}. 
	
	There was a fundamental question, whether monogamy of non-local correlations could survive in RC
	because it already failed in the  CHSH case \cite{ref1}. We  provide a general Theorem stating that
	actually \textit{no} two-party Bell inequality can exhibit a monogamy relation under the constraints (\ref{eq:RCcond}). This implies that: 
	
	\begin{itemize}
		\item 		The non-local correlations between two space-like  separated devices can become completely non-monogamous in relativistic causal theories when the measurement events of the parties are in accordance with the space-time configuration of Fig. \ref{Figure1}. 
		
	\end{itemize}

Consider a general bipartite Bell inequality $G$ of the form
	\begin{eqnarray}
	\label{eq:bip-Bell0}
	G := \sum_{a,b,x,y} V(a,b,x,y) P(a,b|x,y) \leq \omega_c(G),
	\end{eqnarray}
	Here $\omega_c(G)$ denotes the optimal classical value of the left-hand side of the above inequality. The following Proposition shows that in a three-party Bell test with the measurement events occurring in the space-time configuration in Fig. \ref{Figure1}, relativistic causal correlations exist that allow both pairs of parties A-B and B-C to simultaneously observe the maximum relativistic causal value $\omega_{rc}(G)$ of the inequality.
	\begin{prop}[1]
		Consider any bipartite Bell inequality $G$ of the form in Eq. (\ref{eq:bip-Bell0}). Suppose three players perform their measurements in the space-time configuration of Fig.\ref{Figure1}, and that both Alice-Bob and Bob-Charlie test for the violation of $G$. Then, there exist correlations $\{P(a,b,c|x,y,z)\}$ in RC theories that allow both A-B and B-C to achieve $\omega_{rc}(G)$.  
	\end{prop}
	The  proof is provided in     \cite{ref14}. 
	
	Due to the crucial role of monogamy in device-independent security this Proposition immediately rises a  question about security 
	based on the relativistic causality alone. In the next paragraph  we shall answer this question in the negative, 
	showing a coordinated hacking strategy for a group of eavesdroppers which
	breaks the security of {\it any} device-independent protocol based on a violation of {\it any} Bell inequality. 
	
	However, the most important consequence of the above Proposition goes deeply into the very roots of the structure of the correlations in the considered theory. So far, given a composite physical system and its correlations polytope,  the {\it extremal points} of the latter  always {\it guaranteed a lack of correlations of the system with any external observer}. 
	According to the above Theorem, there is a dramatic change in RC theories: 
	{\it virtually all the extremal points} of bipartite RC polytope describe system potentially correlated 
	with some environment which definitely undermines chances for secure information processing.

	\paragraph{A No-go theorem for device-independent security}
	
	Here we present an attack by two eavesdroppers that breaks any device-independent security protocol.
	The hacking strategy is valid for any number of parties and regardless of the Bell test performed by reliable agents. 
	
	\begin{figure}
		
		\begin{centering}
			\includegraphics[width=0.9\columnwidth]{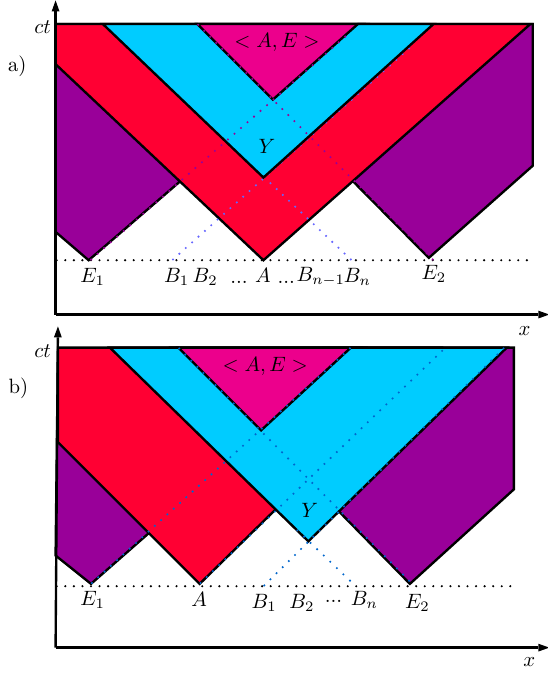}
			\par\end{centering}
		\caption{\label{Figure2} One dimensional illustration of the No-go theorem. Here $E$ stands for the correlation between measurement outputs of the two eavesdropers $E_1$ and $E_2$. The measurement events $A,B_{1},\ldots,B_{N-1}$ of the reliable parties determine a spatial convex hull from which there are two possible situations: a) The party $A$  is in the interior of the convex hull or b) The party $A$ is in te boundary of the convex hull. In both cases the eavesdroppers can choose positions far enough such that   the inputs $Y=y_{1},\ldots y_{N-1}$ can influence the correlation $<A,E>$. This is allowed by RC correlations when the eavesdroppers calibrate the region where information from $<A,E>$  (pink) is accessible only in a region  inside the causal future of all inputs $Y$ (light blue).}
	\end{figure}
	
	We start by considering a multipartite Bell inequality:
	\begin{equation}
	\label{eq:bip-Bell1}
	G\,:=\sum_{\mathbf{r},\mathbf{q}}V(\mathbf{r},\mathbf{q})P\left(\mathbf{r}\mid\mathbf{q}\right)\leq\omega_{c}\left(G\right)
	\end{equation}
	Here, inputs $\mathbf{q}=x,y_{1},\ldots y_{N-1}$ and outputs $\mathbf{r}=a,b_{1},\ldots,b_{N-1}$
	correspond to the devices of the $N$ reliable agents $A,B_{1},\ldots,B_{N-1}$.
	All device-independent security protocols use the violation of a particular Bell inequality
	of the form Eq. (\ref{eq:bip-Bell1}) to ensure the independence of the statistics of any
	eavesdropper from the statistics of the reliable agents. Nevertheless, devices with access to the full set of relativistic causal behaviours, can be correlated with the devices of only two eavesdroppers to reproduce the statistical results of the reliable agents in their entirety and despite maximal violation of the Bell inequality from Eq. (\ref{eq:bip-Bell1}).
	
	\begin{prop}[2]
		Consider any multipartite Bell inequality $G$ of
		the form in Eq. (\ref{eq:bip-Bell1}). Assume that $N$ reliable agents perform
		their measurements in arbitrary space-time positions, in which a violation
		of  $G$,  $\omega^{*}>\omega_{c}$
		is observed. Then, there is a space-time configuration of
		two eavesdroppers' measurements, so that $N$ correlations of those
		measurements will reproduce the statistics of the $N$ reliable agents
		after the eavesdroppers meet or communicate their results. 
		\label{prop:main}
	\end{prop}
	
	An exemplary space-like configuration of the attack  is given in Fig. 1. The proof of the Theorem generalises it to  arbitrary configurations of the reliable agents in 1+3 D space-time [23].

	Configuration obeying relativistic causality in which two eavesdroppers $E_1, E_2$ can use their outputs $c_1, c_2$ to infer the output $a$ of an agent $A$ for any input $x$, is shown in Figure 2.

	A proposal that would seem to be intuitive to restore security is to simply assume that the devices are sufficiently shielded from any influences beyond the usual no-signaling constraints. However, this assumption  is inappropriate in the present investigation for two compelling reasons:

	First, the causal behaviors exploited the attack use point to region influence \cite{ref1, ref14}, which not only is undetectable locally but may be even calibrated by the adversary to be outside of the region occupied by the trusted parties.
	In consequence the shielding from influences can be determined operationally only by the eavesdroppers which renders it a meaningless assumption for the reliable agents.
	
	Second, even more importantly, the point-to-region signaling might be of a physical nature that prevents shielding. We know such a prominent example already: vacuum correlations in quantum electrodynamics is a phenomenon that can not be shielded for the fundamental reasons. 
	
	We are ready to show the main consequence of the above Proposition. We first define the {\it distillable key} $K_{D}^{RC}$. It reads the maximum ratio of the number of key bits divided by the initial number of devices $n$ (in the asymptotic limit) that can be obtained via protocols based on MDLOPC operations from $n$ copies of $P(ABE|XYZ)$. (For details see Definitions $1$ and $2$ in Section D of the Supplemental Material). We will argue now, that $K_D^{RC}$ is zero.
	
	\begin{theorem}[No-go for MDLOPC secure key distribution] For any $P_{AB}\equiv P(A,B|X,Y)$ satisfying non-signaling constraints, there exists a space-time configuration of two eavesdroppers 
		$E_1$ and $E_2$ and a tripartite distribution $P(A,B,E_1,E_2|X,Y)$ satisfying RC constraints, with marginal distribution on $AB$ equal to $P_{AB}$ such that:
		\begin{equation}
		K_D^{RC}(P(ABE_1E_2|XY)) = 0.
		\end{equation}
	\end{theorem}
	
	(For the proof see Appendix D of \cite{ref14}.)

	\paragraph{The Relativistic Causal Polytope}
	The simplest case in which the set of relativistic causal behaviors differs from the set of no-signaling behaviors is the $(3,2,2)$ Bell scenario, according to which three parties perform two binary measurements each.
	The corresponding RC correlations form a polytope that encompasses the usual no-signaling polytope. 
	After our security analysis still, the important general question remains whether there is {\it any } information processing tasks for which the RC correlation work better than the NS ones. 
	For this purpose we provide a complete characterisation of the polytope in terms of the extremal behaviors.
	More specifically, using the software \textit{polymake} \cite{ref23}, we computed \cite{ref23.1} the extremal boxes for the RC polytope in the $(3,2,2)$ scenario \cite{ref14} in the measurement configuration of Fig. \ref{Figure1}. 
	Among the found 153, 600 extremal behaviors, 64 were classical (CL), 2144 no-signaling (NS) and 151,392 relaticistic causal (RC). Considering equivalences up to local transformations and symmetry between Alice and Charlie labs  we got eventually 1 CL, 5 NS and 190 RC equivalence classes shown explicitely in \cite{ref14}) .


	
	We found the dimensionality 
	of the RC polytope in the general $(3,m,n)$ Bell scenario
	(ie. three parties with $m$ measurements of $n$ outcomes each) to be $\mathcal{D}[RC(3,m,n)]= [m(n-1)+1]^{3}+m^{2}(m-1)(n-1)^{2}-1$ (see \cite{ref14}).
	
	
	The above exact characterization of the relativistic causal behaviors  allowed us to reveal the fact that \textit{there are communication complexity scenarios (c.f. \cite{Buhr}), in which some relativistic causal devices outperform all no-signaling devices}. In the Supplementary Material, we show a particular function which in the case of no communication between Alice and Charlie can be guessed by them perfectly by using devices described by RC behaviors, but only with 75  percent chances with NS devices. 
	
	\paragraph{Concluding Remarks}
	
	While the relativistic causality was already known to have unexpected correlation behaviors \cite{ref1}, 
	the lack of monogamy was provided only in one example. 
	
	We succeeded to prove that 
	in the case of bipartite correlations, monogamy is entirely absent - an external environment 
	{\it can always have just a copy of the variable}. This fact has drastic consequences. 
	
	Frist, unlike in no-signaling scenarios, the mutual link between two properties:
	(i) extremality of correlations of a fixed composite system in the corresponding convex set
	and  (ii) lack of correlations with the external environment is broked completely.

	Second, it has given us a hint to prove a security no-go
	for any protocol, based on multipartite Bell inequalities against a coalition of two eavesdroppers  who are in unconstrained space-time positions.
	Shortly speaking we shon that security of key distribution can not be based on relativistic causality alone.

	Going beyond the above results,
	we have provided a full characterisation of the RC correlations in terms of extremal points, which helped us to answer the following question: is there any positive message possible for information processing in the RC scenario? The affirmative answer has been provided: there are cases when RC correlations outperform the no-signaling ones in specific communication complexity reduction problems.

	Let us mention some further potential applications of our results.   The complete characterization of the scenario $(3,2,2)$ in relativistic causal theories provides a useful tool to investigate multi-party Bell non-locality in causal networks, which have been attracting  considerable attention  \cite{causal1,causal2,causal3,causal4,causal5,causal6,ref32}. A further  step would be to investigate the interplay between advantage in communication complexity and hardness in security proofs. 
	Here the following question arises which we leave for forthcoming works : Does an information-theoretic principle such as a non-reduction of communication complexity beyond NS capabilities, ensure the security of protocols? 
	Whether a single eavesdropper can be as powerful as the two also remains an important problem for future research. It seems plausible that the presented attack disallows for {\it any} protocol of secure key distribution. Extending our no-go for all protocols is left as a significant open problem.

	Finally, the present result re-opens the crucial question \textit{whether secure cryptographic protocols can be based on fundamental physical principles only}. These principles are relevant to determine to what extent physics ensures data privacy and randomness. 
	An open question is whether these principles would allow correlations beyond the limits imposed by the non-signaling 
	condition, in which case a new range of phenomena could be studied.

	\section*{Acknowledgements}
	RS acknowledges support of Comision Nacional de Investigacion Ciencia y Tecnologia (CONICYT)  Programa de Formacion Capital Humano Avanzado/Beca de Postdoctorado en el extranjero (BECAS CHILE) 74160002 and John Templeton Foundation, MK, DG and PH acknowledge support of  John Templeton Foundation, KH acknowledges the  grant Sonata Bis 5 (grant number: 2015/18/E/ST2/00327) from the National Science Center and John Templeton Foundation, K. H and P. H. are also supported by the Foundation for Polish Science (IRAP project, ICTQT, contract no.2018/MAB/5, co-financed by EU within Smart Growth Operational Programme, RR acknowledges support of research project "Causality in quantum theory: foundations and applications" of the Foundation Wiener-Anspach and from the Interuniversity Attraction Poles 5 program of the Belgian Science Policy Office under the grant IAP P7-35 photonics@be.", DG acknowledges partial support from Grant FONDECYT Iniciación number 11180474, Chile, and DS acknowledges support of NCN grants 2016/23/N/ST2/02817 and 2014/14/E/ST2/00020  .


\onecolumngrid
	
	\section*{Supplementary Material}

The most fundamental cryptography task is to achieve secure communication between two separated parties - this is the task of secure key distribution. 
We focus on this task in the parallel measurement scenario, as in this case  adversary can not pursue drastic attacks. As one of the main results, we show that contrary to the case of non-signaling theory, {\it there is  no protocol in the parallel measurement scenario, that allows for distributing key secure against RC adversary.}

 
The way to check if security is possible in NS theory, is to test the level of violation of a Bell inequality. Special cases of Bell inequalities with only either $0$ or $1$ coefficients are called games \cite{Bell-nonlocality}. In NS theory if some two party share a device, statistics of which violate a Bell inequality by sufficiently high amount (or in case of games
- win the game with high enough probability), then the so called {\it monogamy } holds : none of them can achieve the same with respect to some other party i.e. win the game with someone with large probability. This fact is  fundamental for secure  communication in QTI an NS theories. On our way to answer the main question we therefore first study if monogamy takes place in RC. Interestingly, we show a drastic violation of this phenomenon in the RC scenario.The above fact leads to our main contribution: That key rate in any Bell violation based security protocol is zero against RC adversaries. 


\section*{Appendix A: General constraints of Relativistic Causal  correlations}
\label{sec:gen-RC}

In this Appendix we introduce a general formalism for the study of
RC constraints in multipartite scenarios and a general space-time. Consider a set of $\left[n\right]=\left\{ 1,\ldots,n\right\} $ parties
with a string of inputs $\boldsymbol{x}=\left\{ x_{1},\ldots,x_{n}\right\} $
and string of outputs $\boldsymbol{a}=\left\{ a_{1},\ldots,a_{n}\right\} $,
$\boldsymbol{S}\subseteq\left[n\right]$ with complement $\boldsymbol{S}^{c}$,
such that $\boldsymbol{a}_{\mathbf{S}}=\left\{ a_{i}\right\} _{i\in\boldsymbol{S}}$
and analogous definition for $\boldsymbol{x}_{\boldsymbol{S}}$. In this scenario,
the usual no-signaling constraints can be written as:

\begin{equation}
P\left(\boldsymbol{a}_{\boldsymbol{S}}\mid\boldsymbol{x}_{\boldsymbol{S}}\right)=\sum_{\boldsymbol{a^{\prime}}_{\boldsymbol{S}^{c}}}P\left(\boldsymbol{a}^{\prime}_{\boldsymbol{S}^{c}},\boldsymbol{a}_{\boldsymbol{S}} \mid\boldsymbol{x}^{\prime}_{\boldsymbol{S}^{c}},\boldsymbol{x}_{\boldsymbol{S}}\right)=\sum_{\boldsymbol{a^{\prime\prime}}_{\boldsymbol{S}^{c}}}P\left(\boldsymbol{a}^{\prime\prime}_{\boldsymbol{S}^{c}},\boldsymbol{a}_{\boldsymbol{S}}\mid\boldsymbol{x}^{\prime\prime}_{\boldsymbol{S}^{c}},\boldsymbol{x}_{\boldsymbol{S}}\right)
\end{equation}
for all $\boldsymbol{x}^{\prime}_{\boldsymbol{S}^{c}}$, $\boldsymbol{x}^{\prime\prime}_{\boldsymbol{S}^{c}}$
In words, these constraints state that the probability distribution of the outputs of any subset of parties is independent from the inputs
of the complementary set of parties. In the multi-partite relativistic
causal set of constraints we also consider the space-time measurement
events $\left\{ M_{a_{1}}^{x_{1}},\ldots,M_{a_{n}}^{x_{n}}\right\} $
in the space-time $\left(\mathcal{M},\,g_{\mu\nu}\right)$ for some
coordinate system (in special relativity this could be a particular
reference frame). For a party $p$ to influence the correlations of a set of parties $\boldsymbol{S}\nsupseteq\left\{ p\right\} $
the event $M_{a_{p}}^{x_{p}}$ must satisfy:

\begin{equation}
\bigcap_{q\in\boldsymbol{S}}J^{+}\left(M_{a_{q}}^{x_{q}}\right)\subset J^{+}\left(M_{a_{p}}^{x_{p}}\right)\label{eq:causalcond}
\end{equation}
In words, this condition states that the causal future $J^{+}\left(M_{a_{p}}^{x_{p}}\right)$ of  party $p$'s measurement event  contains the intersection of the causal futures of the measurement events of all the parties $q \in \boldsymbol{S}$.
Thus, a set $\boldsymbol{K}$ of parties, might signal to another set
$\boldsymbol{S}$ iff for each $\left\{ p\right\} \in\boldsymbol{K}$
the condition (\ref{eq:causalcond}) is satisfied. If $\boldsymbol{K}$
can't signal to $\boldsymbol{S}$ we say $\boldsymbol{K}\nrightarrow\boldsymbol{S}$,
thus the RC conditions are all those of the form: 

\begin{equation}
P\left(\boldsymbol{a}_{\boldsymbol{S}}\mid\boldsymbol{x}_{\boldsymbol{S}}\right)=\sum_{\boldsymbol{a^{\prime}}_{\boldsymbol{S}^{c}}}P\left(\boldsymbol{a}^{\prime}\mid\boldsymbol{x}^{\prime}\right)=\sum_{\boldsymbol{a^{\prime\prime}}_{\boldsymbol{S}^{c}}}P\left(\boldsymbol{a}^{\prime\prime}\mid\boldsymbol{x}^{\prime\prime}\right)\: \textrm{iff} \: \forall  \boldsymbol{K}\subseteq\boldsymbol{S}^{c}, \: \boldsymbol{K}\nrightarrow\boldsymbol{S} \label{eq:RCconstraint}
\end{equation}
 Of course, in general this definition has redundant constraints and in general
a  subset of these constraints can determine the full set. By definition the RC constraints are a subset of no-signaling constraints, therefore no-signaling boxes  satisfy the RC constrains while the opposite
is not always true. An important remark to be made here is that for any spacetime
and spacelike separated parties we have: 
\begin{equation}
P\left(a_{p}\mid x_{p}\right)=\sum_{\boldsymbol{a^{\prime}}_{\{p\}^{c}}}P\left(\boldsymbol{a}^{\prime}\mid\boldsymbol{x}^{\prime}\right)=\sum_{\boldsymbol{a^{\prime\prime}}_{\{p\}^{c}}}P\left(\boldsymbol{a}^{\prime\prime}\mid\boldsymbol{x}^{\prime\prime}\right)\label{eq:minimal}
\end{equation}
for any single party $p$. This is the minimum number of RC constraints,
which corresponds to the largest correlation polytope. Since always
the single party outcome probabilities are well defined, the signaling
in RC can only target  sets of parties with two or more elements, i.e. to a \emph{region}. In this article we only consider cases where signaling from a region is the union of the several individual signals from parties inside that region, accordingly we designate the signaling allowed by RC as  \emph{point to region} (PTR) signaling without any loss of generality. 

\section*{Appendix B: Communication complexity advantage in Relativistic Causal theories.}

\begin{figure}

 \includegraphics[scale=0.9]{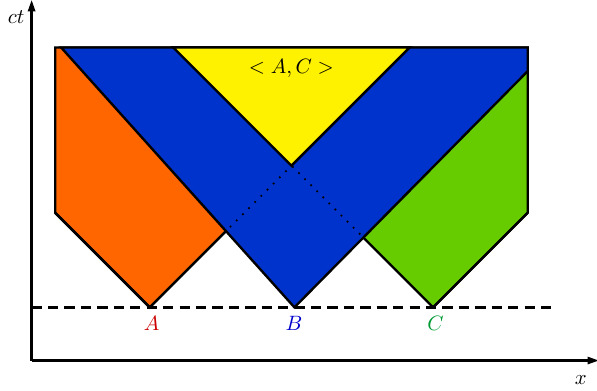}
 \caption{\label{Figure3} A particular spacetime configuration of measurement events in the three-party Bell experiment. The spacetime locations of Alice, Bob and Charlie measurement events are $A$, $B$, $C$ respectively. The yellow area shows $J^{+}\left(A\right)\cap J^{+}\left(C\right)$, the only region where the information from the correlations between the outputs of Alice and Charlie, denoted $<A,C>$,  is accessible. The crucial property of this measurement configuration is that $J^{+}\left(A\right)\cap J^{+}\left(C\right)\subset J^{+}\left(B\right) $.}
 \end{figure}

The relativistic causal correlations in the measurement configuration of  Fig.\ref{Figure3}  are separated from the usual no-signaling correlations by constraints of the form
\begin{eqnarray}
\sum_{b} P(a,b,c|x,y,z) - \sum_{b} P(a,b,c|x,y',z) = 0 \qquad \forall a,c,x,z,y \neq y'.
\end{eqnarray} 
The usual no-signaling constraints impose equality above while this equality is not necessary for relativistic causality to hold as shown in \cite{ref1}. The relaxation of these constraints is also reflected in a difference between the optimal success probability $\omega(G)$ of multi-player games in NS theories versus that in RC theories. We first note that as in the no-signaling case, the calculation of the optimal success probability of multi-player games in RC theories can be achieved in polynomial time by means of a linear program and second we explain how advantage in some of these games imply communication complexity advantages.

As a first example of the difference in $\omega(G)$ between NS and RC theories, consider the Guess-Your-Neighbour's-Input Game (GYNI) in the (3,2,2) Bell scenario. The inputs $x,y,z$ to the three parties in the game obey the promise $x \oplus y \oplus z = 0$ and the task is for each party to output their neighbour's input, so that the expression for the success probability in the game is given by
\begin{equation}
\omega(\text{GYNI})=\frac{1}{4}\left[ P\left(000|000\right)+P\left(110|011\right)+P\left(011|101\right)+P\left(101|110\right)\right] \label{eq:GYNIgame}
\end{equation}
It was shown in \cite{ref18} that $\omega_c(\text{GYNI}) = \omega_q(\text{GYNI}) = \frac{1}{4}$ while correlations obeying the no-signaling constraints allow $\omega_{ns}(\text{GYNI}) = \frac{1}{3}$. Here, $\omega_c, \omega_q$ and $\omega_{ns}$ denote the optimal success probability in classical, quantum and no-signaling theories respectively, while similarly $\omega_{rc}$ will denote the optimal success probability in theories that only impose relativistic causality. 
A simple maximization over the constraints in Eq.(\ref{eq:RCconstraint}) gives that $\omega_{rc}(\text{GYNI}) = \frac{1}{2}$ and this optimal value is achieved by the RC Box (Extremal box class nr. 77
in Appendix G):
\begin{equation}
B_{\text{GYNI}}^{RC}:\:P\left(abc\mid xyz\right)=\begin{cases}
\frac{1}{2}, & \textrm{if}\,\left(1\oplus b\oplus c\oplus y\right)\left(1\oplus a\oplus b\oplus x\right)=1\\
0, & \textrm{otherwise}
\end{cases}
\end{equation}  

As a second example, we present games where RC correlations allow the players to win with certainty (success probability one) while the best no-signaling strategy gives a success probability less than one. In these games, we consider three parties, of whom only the outputs of two parties appear in the winning constraint, while the third player helps the others achieve their task, so that one might term these games as "games with allies" (GWA). Specifically, we propose a GWA game for Alice and Charlie with Bob as the ally, with a winning constraint given by
\begin{equation}
xy\oplus yz=a\oplus c, \label{eq:gwagame}
\end{equation}
where as usual $x,y,z$ denote the inputs of the three players and $a,b,c$ denote their respective outputs. For this game, a simple maximization over the usual no-signaling constraints by a linear program shows that $\omega_{ns}(GWA) = \frac{3}{4}$. In fact, a classical strategy exists that achieves this value, and is simply given when Alice and Charlie output $a=c=0$ for any input $x,y,z$. When $y=0$, this strategy satisfies the winning constraint $a \oplus c = (x \oplus z) y = 0$, and when $y=1$, this strategy satisfies $a \oplus c = (x \oplus z)$ in exactly half of the cases, so that the optimal success probability $\omega_c(GWA) = \frac{3}{4}$ is achieved. On the other hand using a RC box is it possible to win the GWA with
certainty. Specifically, consider the RC Box (Extremal box class nr. 76 in Appendix
G): 
\begin{equation}
B_{\text{GWA}}^{RC}:\:P\left(abc\mid xyz\right)=\begin{cases}
\frac{1}{2}, & \textrm{if}\,\left(1\oplus a\oplus b\oplus xy\right)\left(1\oplus b\oplus c\oplus zy\right)=1\\
0, & \textrm{otherwise}
\end{cases}
\end{equation}
This box satisfies $a \oplus b = xy$ and $b \oplus c = zy$ (two Popescu-Rohrlich type boxes between A-B and B-C) so that it directly satisfies $a \oplus c = xy \oplus zy$, which gives $\omega_{rc}(\text{GWA}) = 1$. In the literature the condition (\ref{eq:gwagame}) appears in \cite{vanDam}
as a communication complexity task for Alice and Charlie: They must compute
functions $f\left(x,y,z\right)=h\left(x,y\right)\oplus g\left(y,z\right)$,
sharing 1 bit of information and without communication with Bob. This shows that RC Boxes can be used to trivialize some communication complexity tasks \cite{vanDam}. This remarkable result, suggest that a communication principle demanding the no-trivialization  of GWA games has direct consequences on RC correlations. Could it be that a communication principle   implies enough restrictions to certify a security protocol in RC theories?  We leave for future research the investigation of this question.


\section*{Appendix C: Lack of monogamy for two-player games in RC theories.}

An important consequence of the relaxation of the no-signaling constraints to those that are sufficient to ensure relativistic causality is the resulting lack of monogamy for general two-player games in RC theories. In particular, when the players' measurements are arranged in the space-time configuration of Fig.\ref{Figure1}, for any two-player game $G$ it holds that $\omega_{rc}(G^{AB}) = \omega_{rc}(G^{BC}) = \omega_{ns}(G)$. In other words, both players are able to achieve the maximum no-signaling (equal to the relativistic causal) value of the two-player game $G$ in this configuration. We give the proof of this statement for a general bipartite Bell inequality in this section. 

Consider a general bipartite Bell inequality $G$ of the form
\begin{eqnarray}
\label{eq:bip-Bell}
G := \sum_{a,b,x,y} \alpha_{a,b,x,y} P(a,b|x,y) \leq \omega_c(G),
\end{eqnarray}
where we take without loss of generality $\alpha_{a,b,x,y} \geq 0$ and normalize the inequality so that $\omega_c(G) \leq 1$. 
\begin{prop}[1]
	Consider any bipartite Bell inequality $G$ of the form in Eq.(\ref{eq:bip-Bell}). Suppose three players perform their measurements in the space-time configuration of Fig.\ref{Figure1}, and that both Alice-Bob and Bob-Charlie test for the violation of $G$. Then, there exist correlations $\{P(a,b,c|x,y,z)\}$ in RC theories that allow both A-B and B-C to achieve $\omega_{ns}(G)$.  
\end{prop}
\begin{proof}
	We construct the required RC box $\{P(a,b,c|x,y,z)\}$ depending on the bipartite Bell inequality $G$ as follows. Let $\{Q(a,b|x,y)\}$ be a two-party no-signaling box that achieves the maximum no-signaling (equal to relativistic causal, in this bipartite case) value $\omega_{ns}(G)$. 
	
	Fix $y=1$. The box $\{Q(a,b|x,y=1)\}$ is local realistic by virtue of the fact that party B only chooses the single input $y=1$. We construct a symmetric extension of $\{Q(a,b|x,y=1)\}$ to the three-party box $\{\tilde{Q}_1(a,b,c|x,y=1,z)\}$ such that the two-party marginals A-B and C-B are equal to $Q(a,b|x,y=1)$, i.e., we impose
	\begin{eqnarray}
	Q(a,b|x,y=1) = \sum_{c} \tilde{Q}_1(a,b,c|x,y=1,z) = \sum_{a'} \tilde{Q}_1(a',b,c'|x',y=1,z') \quad \forall b,a=c',x=z'.
	\end{eqnarray} 
	Such a symmetric extension can always be constructed for the local realistic box $\{Q(a,b|x,y=1)\}$. To make this more explicit, suppose that the box has the following decomposition into classical deterministic boxes
	\begin{eqnarray}
	Q(a,b|x,y=1) = \sum_{\lambda} p_{\lambda} Q_A(a|x,\lambda) Q_B(b|y=1,\lambda).
	\end{eqnarray} 
	One can then construct the symmetric extension $\{\tilde{Q}_1(a,b,c|x,y=1,z)\}$ as
	\begin{eqnarray}
	\tilde{Q}_1(a,b,c|x,y=1,z) = \sum_{\lambda} p_{\lambda} Q_A(a|x,\lambda) Q_B(b|y=1,\lambda) Q_A(c|z,\lambda),
	\end{eqnarray}
	where the marginal distribution for party C is the same as that for A, and $Q_A, Q_B$ are deterministic boxes. Note that the symmetric extension obeys all the usual no-signaling constraints i.e., every bipartite marginal $\tilde{Q}_1(a,b|x,y=1)$ and $\tilde{Q}_1(b,c|y=1,z)$ as well as the single-party marginals $\tilde{Q}_1(a|x), \tilde{Q}_1(b|y=1)$ and $\tilde{Q}_1(c|z)$ are well-defined independent of the inputs of the remaining parties. 
	
	Similarly, fix $y=2,3, \dots |Y|$ and construct the corresponding symmetric extensions $\tilde{Q}_k(a,b,c|x,y=k,z)$ for each of the local realistic boxes $Q(a,b|x,y=k)$. In all these boxes again, the bipartite and single-party marginals are well-defined independent of the inputs of the other parties, and moreover we have that
	\begin{eqnarray}
	\tilde{Q}_k(a|x) = \tilde{Q}_{k'}(a|x) = \sum_b Q(a,b|x,y=1)\quad \forall a,x,k,k' \nonumber \\
	\tilde{Q}_k(c|z) = \tilde{Q}_{k'}(c|z) = \sum_b Q(c,b|z,y=1) \quad \forall c,z,k,k'
	\end{eqnarray}
	by the property of the symmetric extension, i.e., A and C's marginals are the same in each extension. 
	
	Now, putting together all the symmetric extensions, we obtain the combined box $P(a,b,c|x,y,z)$ that is the required box shared by the three parties A,B and C, with $P(a,b,c|x,y=k,z) = \tilde{Q}_k(a,b,c|x,y=k,z)$ for every $k,a,b,c,x,z$. This box satisfies all the RC constraints in Eq.(\ref{eq:RCconstraint}) by the argument above. Note that in general, 
	\begin{eqnarray}
	\sum_b P(a,b,c|x,y=k,z) \neq \sum_b P(a,b,c|x,y=k',z) \quad k \neq k',
	\end{eqnarray}
	but we have seen that this is precisely the missing constraints from the usual no-signaling conditions, that is not necessary to ensure by causality in this measurement configuration. Since the two-party marginals $P(a,b|x,y)$ and $P(c,b|z,y)$ are both equal to $Q(a,b|x,y)$, we have that both A-B and B-C achieve the maximum no-signaling value $\omega_{ns}(G)$. This completes the proof.   
\end{proof}	
As an example of the general proposition above, we find that the following RC box 
\begin{equation}
B_{G_{u}}^{RC}:\:P\left(abc\mid xyz\right)=\begin{cases}
\frac{1}{d}, & \textrm{if}\,a=\pi_{xy}\left(b\right),\,c=\pi_{zy}\left(b\right)\\
0, & \textrm{otherwise}
\end{cases}\label{eq:MonBreak}
\end{equation}
allows both A-B and B-C to achieve the maximum no-signaling value of 1, for any unique game $G_{u}$ defined by a set of permutations $\{\pi_{xy}\}$. 

\section*{Appendix D: The No-go theorem for device-independent security in relativistic causal theories}

In this appendix we complete the proof of the No-go theorem presented in our article. The main theorem is:

\begin{theorem}\label{thm:eves}
	  Assume that $N$ reliable agents perform
their measurements with  inputs $\mathbf{q}=y_{1},\ldots y_{N}$ and outputs $\mathbf{r}=b_{1},\ldots,b_{N}$ in arbitrary spacelike separated positions, to compute any multipartite Bell inequality $G$ of
the form:\begin{equation}
G\,:=\sum_{\mathbf{r},\mathbf{q}}V(\mathbf{r},\mathbf{q})P\left(\mathbf{r}\mid\mathbf{q}\right)\leq\omega_{c}\left(G\right)
\end{equation} in which a violation
of  $G$,  $\omega^{*}>\omega_{c}$
is observed. Then, there is a space-time configuration for two
 eavesdroppers' measurements, so that $N$ correlations of those
measurements will reproduce the statistics of the $N$ reliable agents
after the eavesdroppers meet or communicate their results. 
\end{theorem}

\begin{proof}
We begin with a brief description of the idea of the proof. We consider two eavesdroppers $E_1$ and $E_2$. We show that satisfying RC constraints one can construct a device such that the inputs and outputs of the honest parties' devices are encoded into correlations between $E_1$ and $E_2$. In order to avoid signaling, the local
marginals of the eavesdroppers are uniform,
as one of the Eaves in a sense one-time-pads the information of the other. We borrow this idea
from the simplest secrete sharing scheme. To give a concrete example, the outputs of the honest parties $\mathbf{r}=b_1,\ldots,b_N$ will be encoded into variables of Eves $\mathbf{d}$ and $\mathbf{c}$ respectively as follows. For each of $j \in \{1,\ldots,N\}$ there is:  
$d_j =c_j\oplus_{H_j} b_j$ where addition is 
modulo $H_j$ - the dimension of $b_j$, and distribution of $c_j$ is $\frac{1}{H_j}$. It is easy to
see that none of the eavesdroppers can gain
any knowledge about each of $b_j$, however
upon meeting they can learn each of $b_j$
perfectly.

We are ready to proceed with details of the proof.
 Consider two eavesdroppers $E_{1},E_{2}$  with devices that have only outputs $\mathbf{s}_{1}=c_{1},\ldots,c_{N},z_{1},\ldots z_{N}$
, $\mathbf{s}_{2}=d_{1},\ldots,d_{N},w_{1},\ldots w_{N}$ respectively.
Let's say the eavesdroppers want to attack all trusted parties $B_{1},\ldots,B_{N}$.
Given a particular reference system $S\left(\vec{r},t\right)$ there
always exist an event $p_{B}$ in a causal space-time, such that the
space-time convex hull $\mathcal{B}$ of the spacelike separated measurement
events $p_{B_{1}},\ldots,p_{B_{N}}$ performed by the $B_{1},\ldots,B_{N}$
parties is completely inside its causal past $J^{-}\left(p_{B}\right)\supseteq\mathcal{B}$.
When the two eavesdroppers $E_{1},E_{2}$ can choose any spacelike
separated positions for their measurement events $p_{E_{1}},p_{E_{2}}$,
in particular they can satisfy $J^{+}\left(p_{E_{1}}\right)\cap J^{+}\left(p_{E_{2}}\right)\subseteq J^{+}\left(p_{B}\right)$
for any space-time which is causal and simply connected. In this case
every $B_{1},\ldots,B_{N}$ can signal to any correlation between
the outputs of $E_{1},E_{2}$.

Then, the eavesdroppers could distribute a behavior that satisfies:
\begin{equation}
\sum_{\mathbf{r}}\tilde{Q}_{\mathbf{k}}\left(\mathbf{r},\mathbf{s}_{1},\mathbf{s}_{2}\mid\mathbf{q}=\mathbf{k}\right)\neq\sum_{\mathbf{r}}\tilde{Q}_{\mathbf{k}^{\prime}}\left(\mathbf{r},\mathbf{s}_{1},\mathbf{s}_{2}\mid\mathbf{q}=\mathbf{k}^{\prime}\right)\:\mathbf{k}\neq\mathbf{k}^{\prime}
\end{equation}
The correlation between the outputs of $E_{1},E_{2}$ can be choosen,
such that $d_{j}=c_{j}\bigoplus_{H_{j}}b_{j}$ and $w_{j}=z_{j}\bigoplus_{L_{j}}y_{j}$,
with $H_{j}$ the dimension of outputs $b_{j}$ (also outputs $c_{j},d_{j}$
are choosen to have dimension $H_{j}$), $L_{j}$ the dimension of
inputs $y_{j}$ (also outputs $z_{j},w_{j}$ are choosen to have dimension
$L_{j}$) and $\bigoplus_{H_{j}},\bigoplus_{L_{j}}$ are sums mod
$H_{j}$ and mod $L_{j}$, respectively.

Now, we should check that no-signaling conditions are satisfied according
to the scenario. First, because the $E_{1},E_{2}$ have no input,
they can not signal to the $B_{1},\ldots,B_{N}$. Second, the $\tilde{Q}_{\mathbf{k}}\left(\mathbf{r},\mathbf{s}_{1},\mathbf{s}_{2}\mid\mathbf{q}=\mathbf{k}\right)$
is a classical distribution because it has a single input $\mathbf{k}$
and in consequence we can choose: 
\begin{eqnarray}
\tilde{Q}_{\mathbf{k}}\left(\mathbf{r},\mathbf{s}_{1},\mathbf{s}_{2}\mid\mathbf{q}=\mathbf{k}\right) & = & \sum_{\lambda}p_{\lambda}Q_{B}\left(\mathbf{r}\mid\mathbf{k},\lambda\left(\mathbf{r},\mathbf{k}\right)\right)Q_{E}\left(\mathbf{s}_{1},\mathbf{s}_{2}\mid\mathbf{k},\lambda\left(\mathbf{r},\mathbf{k}\right)\right)\:\forall\mathbf{k}\label{eq:split}
\end{eqnarray}
where $Q_{B}\left(\mathbf{r}\mid\mathbf{k},\lambda\left(\mathbf{r},\mathbf{k}\right)\right)$
reproduce the marginals of $Q\left(\mathbf{r}\mid\mathbf{q}=\mathbf{k}\right)$
for each particular $\mathbf{k}$ when $Q\left(\mathbf{r}\mid\mathbf{q}\right)$
is the no-signaling box that achieves the value $\omega^{\prime}(G)>\omega_{c}\left(G\right)$
expected by the parties $B_{1},\ldots,B_{N}$ for the Bell inequality
$G$. Because $Q\left(\mathbf{r}\mid\mathbf{q}\right)$ is no-signaling,
then no $B_{i}$ signals to any $B_{j}$. Now, what is left is to
check that $B_{1},\ldots,B_{N}$ do not signal neither to $E_{2}$
nor to $E_{1}$. That is: 
\begin{eqnarray}
\sum_{\mathbf{r},\mathbf{s}_{2}}\tilde{Q}_{\mathbf{k}}\left(\mathbf{r},\mathbf{s}_{1},\mathbf{s}_{2}\mid\mathbf{q}=\mathbf{k}\right) & = & \sum_{\mathbf{r},\mathbf{s}_{2}}\tilde{Q}_{\mathbf{k}^{\prime}}\left(\mathbf{r},\mathbf{s}_{1},\mathbf{s}_{2}\mid\mathbf{q}=\mathbf{k}^{\prime}\right)\:\forall\mathbf{k}\neq\mathbf{k}^{\prime}\nonumber \\
\sum_{\mathbf{r},\mathbf{s}_{1}}\tilde{Q}_{\mathbf{k}}\left(\mathbf{r},\mathbf{s}_{1},\mathbf{s}_{2}\mid\mathbf{q}=\mathbf{k}\right) & = & \sum_{\mathbf{r},\mathbf{s}_{1}}\tilde{Q}_{\mathbf{k}^{\prime}}\left(\mathbf{r},\mathbf{s}_{1},\mathbf{s}_{2}\mid\mathbf{q}=\mathbf{k}^{\prime}\right)\:\forall\mathbf{k}\neq\mathbf{k}^{\prime}\label{eq:cond marginal}
\end{eqnarray}
At this point we remark that $\lambda$ carries on the information
from $\mathbf{r}=b_{1},\ldots,b_{N}$ which determine the correlation
of outputs $c_{j},d_{j}$. Now, since $d_{j}=c_{j}\bigoplus_{H_{j}}b_{j}$
and $w_{j}=z_{j}\bigoplus_{L_{j}}y_{j}$ the outputs $c_{j},d_{j},z_{j},w_{j}$
depend only on the inputs and outputs $y_{j},b_{j}$ of agent $B_{j}$.
Because of the functional dependencies above we can rewrite $Q_{E}\left(\mathbf{s}_{1},\mathbf{s}_{2}\mid\mathbf{k},\lambda\left(\mathbf{r},\mathbf{k}\right)\right)$
as: 
\begin{equation}
Q_{E}\left(\mathbf{s}_{1},\mathbf{s}_{2}\mid\mathbf{k},\lambda\left(\mathbf{r},\mathbf{k}\right)\right)=\prod_{j}Q_{E}^{\left(j\right)}\left(c_{j},d_{j},z_{j},w_{j}\mid\lambda\left(b_{j},y_{j}\right),y_{j}\right)
\end{equation}
Here a valid choice for each $Q_{E}^{\left(j\right)}$ is:
\begin{equation}
Q_{E}^{\left(j\right)}\left(c_{j},d_{j},z_{j},w_{j}\mid\lambda\left(b_{j},y_{j}\right),y_{j}\right)=\begin{cases}
\frac{1}{H_{j}L_{j}}, & d_{j}=c_{j}\bigoplus_{H_{j}}b_{j}\textrm{ and }w_{j}=z_{j}\bigoplus_{L_{j}}y_{j}\\
0, & \textrm{otherwise}
\end{cases}
\end{equation}
If we consider the behavior of the form (\ref{eq:split}) to obtain the
marginal of $E_{2}$ :
\begin{eqnarray*}
\sum_{\mathbf{r},\mathbf{s}_{1}}\tilde{Q}_{\mathbf{k}}\left(\mathbf{r},\mathbf{s}_{1},\mathbf{s}_{2}\mid\mathbf{q}=\mathbf{k}\right) & = & \sum_{\mathbf{r},\mathbf{s}_{1}}\sum_{\lambda}p_{\lambda}Q_{B}\left(\mathbf{r}\mid\mathbf{k},\lambda\left(\mathbf{r},\mathbf{k}\right)\right)Q_{E}\left(\mathbf{s}_{1},\mathbf{s}_{2}\mid\mathbf{k},\lambda\left(\mathbf{r},\mathbf{k}\right)\right)\\
 & = & \sum_{\mathbf{r}}\sum_{\lambda}p_{\lambda}Q_{B}\left(\mathbf{r}\mid\mathbf{k},\lambda\left(\mathbf{r},\mathbf{k}\right)\right)\sum_{\mathbf{s}_{1}}Q_{E}\left(\mathbf{s}_{1},\mathbf{s}_{2}\mid\mathbf{k},\lambda\left(\mathbf{r},\mathbf{k}\right)\right)\\
 & = & \sum_{\mathbf{r}}\sum_{\lambda}p_{\lambda}Q_{B}\left(\mathbf{r}\mid\mathbf{k},\lambda\left(\mathbf{r},\mathbf{k}\right)\right)\sum_{c_{N},z_{N}}\cdots\sum_{c_{1},z_{1}}Q_{E}\left(\mathbf{s}_{1},\mathbf{s}_{2}\mid\mathbf{k},\lambda\left(\mathbf{r},\mathbf{k}\right)\right)
\end{eqnarray*}
But, if we sum the distributions $Q_{E}\left(\mathbf{s}_{1},\mathbf{s}_{2}\mid\mathbf{k},\lambda\left(\mathbf{r},\mathbf{k}\right)\right)$
over the components $\left(c_{i},z_{i}\right)$ of $\mathbf{s}_{1}$
we obtain: 
\begin{eqnarray*}
 &  & \sum_{c_{i},z_{i}}Q_{E}\left(\mathbf{s}_{1},\mathbf{s}_{2}\mid\mathbf{k},\lambda\left(\mathbf{r},\mathbf{k}\right)\right)=\sum_{c_{i},z_{i}}\prod_{j}Q_{E}^{\left(j\right)}\left(c_{j},d_{j},z_{j},w_{j}\mid\lambda\left(b_{j},y_{j}\right),y_{j}\right)\\
 & = & \prod_{j\neq i}Q_{E}^{\left(j\right)}\left(c_{j},d_{j},z_{j},w_{j}\mid\lambda\left(b_{j},y_{j}\right),y_{j}\right)\sum_{c_{i},z_{i}}Q_{E}^{\left(i\right)}\left(c_{i},d_{i},z_{i},w_{i}\mid\lambda\left(b_{j},y_{j}\right),y_{j}\right)\\
 & = & \prod_{j\neq i}Q_{E}^{\left(j\right)}\left(c_{j},d_{j},z_{j},w_{j}\mid\lambda\left(b_{j},y_{j}\right),y_{j}\right)\sum_{c_{i},z_{i}}Q_{E}^{\left(i\right)}\left(d_{i}=c_{i}\bigoplus_{H_{i}}b_{i},w_{i}=z_{i}\bigoplus_{L_{i}}y_{i}\mid\lambda\left(b_{j},y_{j}\right),y_{j}\right)\\
 & = & \prod_{j\neq i}Q_{E}^{\left(j\right)}\left(c_{j},d_{j},z_{j},w_{j}\mid\lambda\left(b_{j},y_{j}\right),y_{j}\right)\left(\frac{1}{H_{i}L_{i}}\right)
\end{eqnarray*}
where in the last step we use the fact that the permutations $\pi_{b_{j}}\left(\cdot\right)=\left(\cdot\right)\bigoplus_{H_{j}}b_{j}$
and $\pi_{y_{j}}\left(\cdot\right)=\left(\cdot\right)\bigoplus_{L_{j}}y_{j}$
have a unique value. Since the above calculation is equally valid
when summing up over every pair $\left(c_{i},z_{i}\right)$ of $\mathbf{s}_{1}$
we have:
\begin{eqnarray*}
\sum_{\mathbf{r},\mathbf{s}_{1}}\tilde{Q}_{\mathbf{k}}\left(\mathbf{r},\mathbf{s}_{1},\mathbf{s}_{2}\mid\mathbf{q}=\mathbf{k}\right) & = & \sum_{\mathbf{r}}\sum_{\lambda}p_{\lambda}Q_{B}\left(\mathbf{r}\mid\mathbf{k},\lambda\left(\mathbf{r},\mathbf{k}\right)\right)\sum_{c_{N},z_{N}}\cdots\sum_{c_{1},z_{1}}Q_{E}\left(\mathbf{s}_{1},\mathbf{s}_{2}\mid\mathbf{k},\lambda\left(\mathbf{r},\mathbf{k}\right)\right)\\
 & = & \sum_{\mathbf{r}}\sum_{\lambda}p_{\lambda}Q_{B}\left(\mathbf{r}\mid\mathbf{k},\lambda\left(\mathbf{r},\mathbf{k}\right)\right)\left(\frac{1}{H_{1}L_{1}}\right)\sum_{c_{N},z_{N}}\cdots\sum_{c_{2},z_{2}}\prod_{j\neq1}Q_{E}^{\left(j\right)}\left(c_{j},d_{j},z_{j},w_{j}\mid\lambda\left(b_{j},y_{j}\right),y_{j}\right)\\
 & \vdots\\
 & = & \sum_{\mathbf{r}}\sum_{\lambda}p_{\lambda}Q_{B}\left(\mathbf{r}\mid\mathbf{k},\lambda\left(\mathbf{r},\mathbf{k}\right)\right)\left(\prod_{j=1}^{N}\frac{1}{H_{j}L_{j}}\right)\\
 & = & \left(\prod_{j=1}^{N}\frac{1}{H_{j}L_{j}}\right)\sum_{\mathbf{r}}\sum_{\lambda}p_{\lambda}Q_{B}\left(\mathbf{r}\mid\mathbf{k},\lambda\left(\mathbf{r},\mathbf{k}\right)\right)\\
 & = & \prod_{j=1}^{N}\frac{1}{H_{j}L_{j}}
\end{eqnarray*}
Hence, the marginal of $E_{2}$ is:
\begin{equation}
\sum_{\mathbf{r},\mathbf{s}_{1}}\tilde{Q}_{\mathbf{k}}\left(\mathbf{r},\mathbf{s}_{1},\mathbf{s}_{2}\mid\mathbf{q}=\mathbf{k}\right)=\prod_{j=1}^{N}\frac{1}{H_{j}L_{j}}\:\forall\mathbf{k}
\end{equation}
Now, since the permutations $\pi_{b_{j}}\left(\cdot\right),\pi_{y_{j}}\left(\cdot\right)$
have unique inverses $\pi_{b_{j}}^{-1}\left(\cdot\right),\pi_{y_{j}}^{-1}\left(\cdot\right)$
respectively, we can apply the same arguments when summing up with
every pair $\left(d_{i},w_{i}\right)$ of $\mathbf{s}_{2}$. Then,
a direct calculation shows that the marginal of $E_{1}$ is: 
\begin{equation}
\sum_{\mathbf{r},\mathbf{s}_{2}}\tilde{Q}_{\mathbf{k}}\left(\mathbf{r},\mathbf{s}_{1},\mathbf{s}_{2}\mid\mathbf{q}=\mathbf{k}\right)=\prod_{j=1}^{N}\frac{1}{H_{j}L_{j}}\:\forall\mathbf{k}
\end{equation}
This demonstrates that the marginals of $E_{1}$and $E_{2}$ are independent
from the inputs of $B_{j}$ for each $j\in\{1,\ldots,N\}$.

To complete the attack, we specify how the eavesdroppers can extract
the information of $Q\left(\mathbf{r}\mid\mathbf{q}\right)$ from
$\mathbf{s}_{1},\mathbf{s}_{2}$. As we have seen the value of $Q_{E}\left(\mathbf{s}_{1},\mathbf{s}_{2}\mid\mathbf{q}=\mathbf{k},\lambda\left(\mathbf{r},\mathbf{k}\right)\right)$
is non zero only when $d_{j}=c_{j}\bigoplus_{H_{j}}b_{j}$ and $w_{j}=z_{j}\bigoplus_{L_{j}}y_{j}$
for every $j\in\left\{ 1,...,N\right\} $. Then, from the table of
values $\mathbf{s}_{1},\mathbf{s}_{2}$ is possible to compute a table
$\mathbf{r},\mathbf{q}$ and determine a distribution $Q_{E}\left(\mathbf{r},\mathbf{q}\right)$.
From here we compute:
\begin{equation}
Q\left(\mathbf{r}\mid\mathbf{q}\right)=\frac{Q_{E}\left(\mathbf{r},\mathbf{q}\right)}{\sum_{\mathbf{r}}Q_{E}\left(\mathbf{r},\mathbf{q}\right)}
\end{equation}
Finally, the eavesdroppers are able to compute $Q\left(\mathbf{r}\mid\mathbf{q}\right)$
without affecting the violation $\omega^{\prime}(G)>\omega_{c}\left(G\right)$
observed by parties $B_{j}$. 
\end{proof}
We remark that such attack is possible
because behaviors $\tilde{Q}_{\mathbf{k}}\left(\mathbf{r},\mathbf{s}_{1},\mathbf{s}_{2}\mid\mathbf{q}=\mathbf{k}\right)$
are allowed by the relativistic causal constraints.

\section*{Appendix D: No secure key distillation via direct measurement and LOPC operations, against RC adversaries }

In the previous Section, we have shown that two collaborating eavesdroppers can learn a copy of correlations shared by two honest parties. Intuitive as it is, in such a case, no cryptographic protocol based on these correlations could be accomplished.
However, cryptography is a domain which studies a plethora of security scenarios. Proving a no-go result for each
of them is a difficult task, as the proof is highly dependent on the mathematical description of a particular scenario
(such as two-party cryptographic protocols, secret sharing, anonymous voting, public-key cryptographic protocols,
or private randomness generation). The most fundamental among those scenarios is, no doubt, the secure key distribution between two honest parties against an adversary. To exemplify that it may not be possible in RC, we prove in detail that a broad class of protocols yield zero key rate in the latter scenario.
These are protocols that obtain key via the same measurement in each run of the protocol. They are called {\it Measured device followed by Local Operations and Public Communications} (MDLOPC). Notably,
all modern protocols in device-independent cryptography and quantum device-independent cryptography are MDLOPC operations (see \cite{Bell-nonlocality}  and \cite{NSDI} and references therein). Moreover, in the scenario of secure key distribution against the non-signaling adversary, it is believed that more general class can not yield positive key \cite{Rotem-Sha,Salwey-Wolf}. This fact justifies our focus on MDLOPC operations that lead to positive key in the case of non-signaling adversary \cite{masanes-2009-102,hanggi-2009,Hanggi-phd}. As we will see, no such protocol can achieve a positive key rate against the relativistic causal ones.
Since we will base on the results of Theorem \ref{thm:eves}, we will consider two collaborating adversaries (eavesdroppers) rather than a single one.

\subsection{Scenario for secure key distribution against the relativistic causal adversary.}

In the scenario of secure key distribution against relativistic causal (RC) adversaries, the $M$ 
honest parties share $N$ copies of a (single-use) device. The $M$ parties first measure each of $N$ devices   $P(ABE_1E_2|Y_1,...,Y_M,Z_1,Z_2)$ with {\it the same} direct $\mathbf{q}:=(y_1^1,\ldots,y_M^1)$. They further apply an LOPC (Local Operations and Public Communications) operation on outputs $\mathbf{r}:=(b_1,\ldots,b_M)$ of the measurement. 
This class of operations (introduced in \cite{NSDI}) is called MDLOPC (Measurement on Device followed by LOPC operations).

In practical protocols, there are two phases: testing and key generation.
The measurements in the protocol are taken randomly for both tests and key generation.
There is a finite set of test measurements, while there is a single measurement for key generation (here
$(y^1_i)_{i=1}^M$). The testing rounds are necessary for checking the value of Bell inequality.
If this value is high enough, the data from key generation rounds are processed to produce key. The whole protocol
is aborted otherwise. In what follows, we assume that the device has passed the test, which means that the tested Bell violation is high enough (or even maximal possible). This fact ensures that in the non-signaling case, Alice and Bob would be able in principle to produce key by post-processing (information reconciliation and privacy amplification).
For the sake of clarity, we will
present the proof for $M=2$ honest parties and
later show how to generalize the result for an arbitrary number of them based Ref. \cite{multiparty-squashed}. Consequently, instead of inputs $Y_1,Y_2$ and outputs $B_1,B_2$ we will write $X,Y$ and $A,B$ respectively and use lower case for the values of random variables (e.g. $X=x,Y=y$).
In what follows, the attack by Eves will be
chosen such that $Z_1,Z_2$ will be both unary,
and hence omitted in notation in most cases.

We are ready to define the protocol of key distillation for the case of the two honest parties $A$ and $B$.

\begin{dfn}
	A protocol of key distillation is a sequence of MDLOPC operations $\Lambda=\left\{\Lambda_N\right\}$, performed by the honest parties, each element of which  consisting of a measurement stage $\{{\cal M}=(y^1_i)_{i=1}^M\}$ with $y^1_i = y^1$, followed by a post-processing $\{{\cal P}_N\}$. Moreover, for each consecutive $N$ copies of shared devices $P\equiv P(A,B,E|XYZ)^{\otimes N}$, it outputs a conditional probability distribution such that:
	\begin{eqnarray}
	\label{pdit}
	{||\Lambda_N \left(P^{\otimes N}\right)- \hat{P}_\mathrm{ideal}^{(d_N)} ||}_\mathrm{RC} \le \varepsilon_N \stackrel{N\rightarrow \infty}{\longrightarrow} 0,
	\label{eq:sec-cond}
	\end{eqnarray}
	where an {\it ideal distribution} $\hat{P}_\mathrm{ideal}^{(d_N)}$ is perfectly correlated between the honest parties,
	and product with the device of the eavesdroppers: \begin{equation}
	\hat{P}^{(d_N)}_{\mathrm{ideal}}(A=a,B=b,E_1,E_2|Z_1,Z_2) = \left(\frac{\delta_{A=a,B=b}}{ d_N}\right)\sum_{a,b}P(A=a,B=b,E_1,E_2|Z_1,Z_2),
	\end{equation}
	with $P(A=a,B=b,E_1,E_2|Z_1,Z_2)\equiv \Lambda_N\left(P^{\otimes N}\right)$.
	Moreover by $||P - Q||_{RC} := \sup_{\theta \in RC} ||\theta(P)- \theta(Q)||_1$, we mean the supremum of distinguishability between the distributions achievable by the linear operations satisfying relativistic causality, and $d_N= \mathrm{dim A}^N$.
\end{dfn}

Knowing what the protocols of key distillation in the considered scenario are, we can pass to define the quantity of the key secure against RC adversaries. We limit here ourselves to the case of the key distilled by MDLOPC protocols. 

\begin{dfn} (Key secure against RC adversary)\label{def:key_rate}
	Given a tripartite device $P\equiv P(ABE|XYZ)$ the secret key rate of the protocol of key distillation  $\{\Lambda_N\}$, on $N$ iid copies of the device, denoted by $\mathcal{R}\left(\left.\Lambda\right|_P\right)$ is a number $\limsup_{N\rightarrow \infty} \frac{\log d_N}{
		N}
	$, where $\log d_N$ is the length of a secret key shared between Alice and Bob, with $d_N=\mathrm{dim}_\mathrm{A}\left(\Lambda_N \left(P^{\otimes N}\right)\right)$. The  rate of device independent key secure against RC adversary in the {\it iid} scenario is given by
	\begin{equation}
	K_{DI}^{RC}(P)=\sup_{\Lambda_N \in MDLOPC} {\cal R}\left(\left.\Lambda\right|_P\right),
	\end{equation}
	where the supremum is taken with respect to MDLOPC  protocols.
\end{dfn}

\subsection{No-go for MDLOPC protocols}
To show that the key rate obtained by MDLOPC operations secure against RC adversaries is zero, we demonstrate an upper bound on the key rate and show that it is zero. 
We achieve this task by relating the introduced scenario of security against relativistic causal adversaries with
the so-called {\it secure key agreement} (SKA) \cite{Maurer93,renner-wolf-gap}.

Since we are going to refer to SKA, we recall it briefly here. There, the honest parties and
an eavesdropper share (asymptotically growing number) $N$ copies of a joined probability distribution $P(A,B E)$. The parties
can perform an LOPC operations. The eavesdropper collects the public communication during the protocol. Original security condition that is
demanded for an output of a key distillation protocol is rather involved \cite{Maurer93}. It has been
however shown in \cite{NSDI} that a simple lower bound holds:

\begin{theorem}[\cite{NSDI}]\label{thm:SKA_bound} The secret key rate ${S} ( A : B||E)$ of SKA cryptographic model \cite{CsisarKorner_key_agreement,Maurer93} 
	{is lower bounded by the following asymptotic expression}
	\label{cor:equvialence}
	\begin{equation}\label{eq:equality}
	\mathrm{S} ( A : B||E)_{P(ABE)} \ge
	\sup_{\cal P} \limsup_{N\rightarrow \infty} \frac{\log \mathrm{dim}_\mathrm{A} \left( \mathcal{P}_N\left({P}^{\otimes N}\left(ABE\right)\right)\right)}{N},
	\end{equation}
	{with security condition}
	\begin{equation}
	||{\mathcal{P}_N\left({P}^{\otimes N}\left(ABE\right)\right) - P_\mathrm{ideal}^{(d_N)}}||_1 \le \delta_N \stackrel{N\to \infty}{\longrightarrow}0,
	\label{eqn:normAP}
	\end{equation}
	where $\mathcal{P}=\cup_{N=1}^\infty \{{\mathcal{P}}_N\}$ is a cryptographic protocols consisting of LOPC operations, acting on $N$ iid copies of the classical probability distribution $P(ABE)$. Moreover $P_\mathrm{ideal}^{(d_N)} = \frac{\delta_{A=a,B=b}}{d_N}P(E)$, and $P(E)=\sum_{a,b}P(A=a,B=b,E)$.
\end{theorem}

Let us describe the idea of the proof of the no-go briefly. 
We consider a family of tripartite devices with unary input on the eavesdropper's part (hence omitted in notation) that realize the attack described in Theorem \ref{thm:eves}. For a fixed number of copies $N$, it reads  $P(ABE_1E_2|XY)^{\otimes N}$. 
Since the honest parties first measure their device,  the figure of merit is, in fact, a joined probability distribution $P(ABE_1E_2|X=x,Y=y)^{\otimes N}$. In this case, the norm $||.||_{RC}$ of the difference of two conditional distributions in Eq. (\ref{eq:sec-cond}), is equal to the variational distance between two distributions. Hence, the key secret against the Eves under this particular strategy turns to be upper bounded by the key obtained from $P(AB\bar{E})^{\otimes N}$ by LOPC operations, where ${\bar E}=(AB)$. Indeed from Theorem \ref{thm:eves},
the two Eves can upon meeting
learn the realization of the marginal $P(AB|X=x,Y=y)$. In this way, the Eves switch from the RC scenario to  the secrete key agreement scenario. In the latter scenario, there is a well known bound
on the secure key $S(A:B||\bar{E})$, 
called {\it intrinsic information}. The intrinsic information of a distribution 
$P(ABE)$ is $I(A:B\downarrow E)_{P(ABE)}:=\inf_{\Lambda_E: E\rightarrow E'} I(A:B| E')_{P(ABE')}$. Here $I(A:B|E')_{P(ABE')}$ is the {\it conditional mutual information} equal to $H(AE')+H(BE')-H(E')-H(ABE')$ with $H(X)$ denoting a Shannon entropy of the random variable $X$, and the infimum is taken over stochastic maps transforming $E$ in to $E'$. We have then
\begin{theorem}[\cite{MaurerWolf00CK}] For any tripartite distribution $P(ABE)$, there is:
	\begin{equation}
	S(A:B||E)_{P(ABE)}\leq I(A:B\downarrow E)_{P(ABE)}.
	\end{equation}
	\label{thm:Maurer-Wolf}
\end{theorem}
We are ready now to state the main result of this section - a no go for distillation via MDLOPC operations.

\begin{theorem}[No-go for MDLOPC secure key distribution] For any $P_{AB}\equiv P(A,B|X,Y)$ satisfying non-signaling constraints, there exists a space-time configuration of two eavesdroppers 
	$E_1$ and $E_2$ and a tripartite distribution $P(A,B,E_1,E_2|X,Y)$ satisfying RC constraints, with marginal distribution on $AB$ equal to $P_{AB}$ such that:
	\begin{equation}
	K_D^{RC}(P(ABE_1E_2|XY)) = 0.
	\end{equation}
	\label{thm:no-go-for-RCkey}
\end{theorem}

\begin{proof}
	
	Let us fix  $\eta>0$. For this $\eta$ there exist natural $N$,  $\epsilon_N>0$ and the operation of the MDLOPC protocol, which is  $\eta$-optimal. We denote this operation as ${\Lambda}_N:={\cal P}_N\circ{\cal M}_N$. The first part ${\cal M}_N$ is equivalent to an $N$-fold measurement $x^1,y^1$ (the same on each of the copy of $P(ABE_1E_2|X,Y)$). By $\eta$-optimality we mean that the rate of protocol $\{\Lambda_N\}$ is close by $\eta$ to the optimal $K_D^{RC}(P(ABE_1E_2|XY))$:
	\begin{equation}
	({1/N})\log \mathrm{dim}_A({\Lambda}_N(P^{\otimes N}(ABE_1E_2|X,Y))) \geq K_D^{RC}(P(ABE_1E_2|X,Y)) - \eta,
	\label{eq:eta-optimal}
	\end{equation}
	and
	\begin{equation}
	||{\Lambda}_N(P^{\otimes N}(ABE_1E_2|XY)) - \hat{P}_{\mathrm{ideal}}^{(d_N)}||_{RC} \leq \epsilon_N.
	\end{equation}
	
	Now, thanks to Theorem \ref{thm:eves} the device $P(ABE_1E_2|XY)$ can be chosen such,  that the two Eves, upon meeting are able to learn a copy of a realization of each copy of the distribution $P(A,B|X=x^1,Y=y^1)$. Let us note here, that the Eves can learn not only the outputs $AB$, but also the inputs $X,Y$. However in the class of MDLOPC protocols the measurement $(x^1,y^1)$ that attains supremum in definition of $K_D^{RC}$, is known to Eve(s). This is because  the protocol, as it is usually assumed, is publicly known in particular to adversary. We focus then, on the fact that Eves learn the outputs, so that  $\Omega_E((E_1E_2)^{\otimes N}) = (AB)^{\otimes N}$.
	Since $\Omega_E$ is (in principle unnecessary) action of Eves, the key can only be higher after performing $\Omega_E$:
	\begin{eqnarray}
	\frac{\log \mathrm{dim}_A({\Lambda}_N(P^{(N)}(ABE_1E_2|XY)))}{N}=
	\frac{\log \mathrm{dim}_A({\cal P}_N(P^{(N)}(ABE_1E_2|X=x^1,Y=y^1)))}{N}\leq \nonumber\\ \frac{\log \mathrm{dim}_A({\cal P}_N(P(AB(AB))^{\otimes N}))}{N},
	\label{eq:eta-close}
	\end{eqnarray}
	where $(\frac{1}{N}){\log \mathrm{dim}_A({\cal P}_N(P(AB(AB))^{\otimes N}))}$ is the key rate of the LOPC protocol ${\cal P}_N$ when acting on $P(AB(AB))^{\otimes N}$.
	We have also: 
	\begin{equation}
	||{\Lambda}_N(P(ABE_1E_2|XY)^{\otimes N}) - \hat{P}_{\mathrm{ideal}}^{(d_N)}||_{RC} \leq \epsilon_N \implies
	||{\cal P}_N(P(AB(AB))^{\otimes N}) - P_{\mathrm{ideal}}^{(d_N)}||_{1} \leq \epsilon_N.
	\label{eq:norms-same}
	\end{equation}
	Indeed, measurement $(X,Y)=(x^1,y^1)$ operation composed with an LOPC protocol ${\cal P}_N$ is one of the linear operations satisfying RC, hence we have:
	\begin{eqnarray}
	\sup_{\theta \in RC}||\theta[{\Lambda}_N(P(AB(AB)|XY)^{\otimes N}) - \hat{P}_{\mathrm{ideal}}^{(d_N)}]||_{1} \leq \epsilon_N \implies
	||{\cal P}_N(P(ABE_1E_2|X=x^1,Y=y^1)^{\otimes N}) - \tilde{P}_{\mathrm{ideal}}^{(d_N)}||_{1} \leq \epsilon_N,
	\end{eqnarray}
	where $\tilde{P}_{\mathrm{ideal}}^{(d_N)} =\frac{\delta{A=a,B=b}}{d_N}\sum_{a,b}P(A=a,B=b,E_1E_2|X=x^1,Y=y^1)$. We use now contractivity of the $||.||_1$ norm under stochastic maps, including $\Omega_E$ which maps $E_1E_2$ to $AB$, to obtain finally (\ref{eq:norms-same}) with $P^{(d_N)}_{\mathrm{ideal}}=\frac{\delta(A=a,B=b)}{ d_N} P(AB)$.
	Now $P(AB(AB)|X=x^1,Y=y^1)^{\otimes N}$ is a tripartite probability distribution which we denote as $P(AB(AB))^{\otimes N}$. It is an instance of SKA scenario.
	We can therefore apply the Theorem \ref{thm:SKA_bound} (with $\delta_N = \epsilon_N)$.
	Indeed, in ${\cal P}_N$ we recognize an LOPC operation such that
	(since $\epsilon_N$ can be arbitrarily small)
	we have:
	\begin{equation}
	S(A:B||(AB)) \geq \frac{\log \mathrm{dim}_A({\cal P}_N(P(AB(AB))^{\otimes N}))}{N}.
	\end{equation}
	Now by Eq. (\ref{eq:eta-close}) and (\ref{eq:eta-optimal}) there is:
	\begin{eqnarray}
	S(AB||(AB)) \geq \frac{\log \mathrm{dim}_A({\cal P}_N(P(AB(AB))^{\otimes N}))}{N} &\geq& \nonumber\\
	\frac{\log \mathrm{dim}_A({\Lambda}_N(P^{(N)}(ABE_1E_2|XY)))}{N}&\geq&  K_D^{RC}(P(ABE_1E_2|XY)) -\eta.
	\label{eq:SKAmajorizes}
	\end{eqnarray}
	Since $\eta$ was arbitrary, we can set it to $0$, keeping the above inequality true.
	It is enough to observe now that 
	\begin{equation}
	S(A:B||(AB))_{P(AB(AB))} \leq I(A:B\downarrow (AB))_{P(AB(AB))} = 0
	\label{eq:SKAzero}
	\end{equation}
	The inequality is thanks to Theorem \ref{thm:Maurer-Wolf}.
	Equality holds due to the fact that the intrinsic information $I(A:B\downarrow (AB))$ equals $0$. Indeed, there is 
	\begin{eqnarray}I(A:B|(AB))= H(A(AB)) + H(B(AB)) -H(AB) - H(AB(AB)) = \nonumber\\ H(AAB) + H(BAB) -H(AB) - H(ABAB) = H(AB) + H(BA) - H(AB) - H(AB) =0,\end{eqnarray} and so  $I(A:B\downarrow AB)=\inf_{\kappa:AB\rightarrow E'} I(A:B|E')=0$, as infimum is achieved for $\kappa$ being an identity operation. From Eq. (\ref{eq:SKAmajorizes}), and Eq. (\ref{eq:SKAzero}),  we conclude that
	\begin{equation}
	K_D^{RC}(P(ABE_1E_2|XY)) \leq S(A:B||(AB))\leq 0.
	\end{equation}
	By definition $K_D^{RC} \geq 0$ as the rate $0$ is achieved for a protocol which traces out the input yielding output with $d_N=1$. Hence the assertion follows from the above inequality.
\end{proof}

In the above proof we have considered $M=2$ of the honest parties. We argue now, that analogous result holds for the conference key obtained by the $M>2$ parties, secure against RC adversaries. First, the analogue of a technical Theorem \ref{thm:SKA_bound} of \cite{NSDI} is straightforward. Then, the proof of an analog of Theorem \ref{thm:no-go-for-RCkey} goes along
similar lines as for $M=2$, with a modification in Eq. (\ref{eq:SKAzero}). There we base on the following analogue of Theorem \ref{thm:Maurer-Wolf} shown in \cite{multisquash} (see Theorem 4, and Example 2 there):
\begin{equation}
S(B_1:B_2:\ldots :B_M||E)_{P(B_1,\ldots,B_ME)}\leq \frac{1}{ (M-1)}I(B_1:B_2:\ldots :B_M \downarrow E)_{P(B_1,\ldots,B_ME)},
\end{equation}
for any $M+1$-partite device $P(B_1,\ldots,B_ME)$, where $S(B_1:B_2:\ldots :B_M||E)$ denotes the so called {\it conference key}, while  $I(B_1:B_2:\ldots :B_M \downarrow E) = \sum_{i=1}^M H(B_i,E) - H(B_1,\ldots, B_M,E) - (M-1) H(E)$.
The fact that $I(B_1,\ldots,B_M\downarrow E)$ equals zero for $E=(AB)$ can be checked by direct inspection.

\section*{Appendix E: Dimensionality of the RC polytope}

In this appendix we compute the dimensionality $\mathcal{D}\left[\ldots\right]$
of the polytope of RC correlations in the three party $m$ inputs, $n$ outputs $(3,m,n)$ scenario.
We proceed with our calculation in three steps: 1) begin with the general
set of constraints and divide them in appropriate subsets, 2) compute
in detail the dimensionality of the $(3,2,2)$ scenario (i.e. $\mathcal{D}\left[ RC\left(3,2,2\right)\right]$)
and 3) reproduce computation in 2) for the general scenario of $(3,m,n)$
with the corresponding alterations. 

\subsection*{Step 1: General Setting}

The general setting corresponds to the 3 party, m inputs, n outputs $(3,m,n)$ scenario with correlations satisfying the following constraints :   \begin{eqnarray} 	 P(a,b,c|x,y,z)  & \in & [0,1] \hspace{0.2cm} \forall_{x,y,z,a,b,c}   \label{eq:1.1} \\ 	 \sum_{a,b,c} P(a,b,c|x,y,z) & = & 1 \hspace{0.75cm}\forall_{x,y,z}   \label{eq:1.2} \\   	P(b,c|y,z) = \sum_aP(a,b,c|x,y,z) & = &  \sum_aP(a,b,c|x',y,z) \hspace{0.3cm} \forall_{x,x',y,z,b,c}   \label{eq:1.3}\\ 	P(a,b|x,y) = \sum_cP(a,b,c|x,y,z) & = &  \sum_cP(a,b,c|x,y,z') \hspace{0.3cm} \forall_{z,z',x,y,a,b}   \label{eq:1.4}\\ 	P(a|x) = \sum_{b,c}P(a,b,c|x,y,z) & = &  \sum_{b,c}P(a,b,c|x,y',z') \hspace{0.1cm} \forall_{y,y',z,z',x,a}   \label{eq:1.5}\\ 	P(c|z) = \sum_{a,b}P(a,b,c|x,y,z) & = &  \sum_{a,b}P(a,b,c|x',y',z) \hspace{0.1cm} \forall_{x,x',y,y',z,c}   \label{eq:1.6}   \end{eqnarray} We divide the equalities (\ref{eq:1.2})-(\ref{eq:1.6}) into three sets of constraints $\mathcal{N}=\{(\ref{eq:1.2})\}$, $\mathcal{P}=\{(\ref{eq:1.3}),(\ref{eq:1.4})\}$ and $\mathcal{RC}=\{(\ref{eq:1.5}),(\ref{eq:1.6})\}$. The cardinalities of these sets, for any $m,n$, are given by:   \begin{eqnarray} 	|\mathcal{N}| & = & m^3 \\   \label{eq:1.7} 	|\mathcal{P}| & = & 2 m^2n^2(m-1) \\   \label{eq:1.8} 	|\mathcal{RC}| & = & 2mn(mn-1),   \label{eq:1.9}   \end{eqnarray} and together fully describe the $(3,m,n)$ RC polytope. \\ 

Since the set of normalization constraints $\mathcal{N}$ involves
mutually independent equalities we consider them - without loss of
generality- as independent and describe the dependencies of equations
in other sets with respect to them. 

\subsection*{Step 2: Computing $\mathcal{D}\left[RC\left(3,2,2\right)\right]$}

Here we discuss in detail mutual dependencies between equalities in and between the sets $\mathcal{N}$, $\mathcal{P}$ and $\mathcal{RC}$ for the (3,2,2) scenario. We begin by writing explicitly  all equations of $\mathcal{P}$ and $\mathcal{RC}$ in the form of tables:   \begin{equation} 	\begin{array}{c|cccc|cccc} 	\mathcal{P} && c1 & & c2 && c3 & & c4 \\\hline 	 & \textbf{Q1} & &&& \textbf{Q5} &&& \\ 	  r1 &&  \sum_a P(a00|000) & = & \sum_a P(a00|100) && \sum_c P(00c|000) & = & \sum_c P(00c|001) \\ 	  r2 &&  \sum_a P(a01|000) & = & \sum_a P(a01|100) && \sum_c P(01c|000) & = & \sum_c P(01c|001) \\ 	  r3 &&  \sum_a P(a10|000) & = & \sum_a P(a10|100) && \sum_c P(10c|000) & = & \sum_c P(10c|001) \\ 	  r4 &&  \sum_a P(a11|000) & = & \sum_a P(a11|100) && \sum_c P(11c|000) & = & \sum_c P(11c|001) \\\hline 	    & \textbf{Q2} & &&& \textbf{Q6} &&&  \\ 	  r5 &&  \sum_a P(a00|001) & = & \sum_a P(a00|101) && \sum_c P(00c|010) & = & \sum_c P(00c|011) \\ 	  r6 &&  \sum_a P(a01|001) & = & \sum_a P(a01|101) && \sum_c P(01c|010) & = & \sum_c P(01c|011) \\ 	  r7 &&  \sum_a P(a10|001) & = & \sum_a P(a10|101) && \sum_c P(10c|010) & = & \sum_c P(10c|011) \\ 	  r8 &&  \sum_a P(a11|001) & = & \sum_a P(a11|101) && \sum_c P(11c|010) & = & \sum_c P(11c|011) \\\hline 	    & \textbf{Q3} & &&& \textbf{Q7} &&&  \\ 	  r9 && \sum_a P(a00|010) & = & \sum_a P(a00|110) && \sum_c P(00c|100) & = & \sum_c P(00c|101) \\ 	  r10 &&  \sum_a P(a01|010) & = & \sum_a P(a01|110) && \sum_c P(01c|100) & = & \sum_c P(01c|101) \\ 	  r11 &&  \sum_a P(a10|010) & = & \sum_a P(a10|110) && \sum_c P(10c|100) & = & \sum_c P(10c|101) \\ 	  r12 &&  \sum_a P(a11|010) & = & \sum_a P(a11|110) && \sum_c P(11c|100) & = & \sum_c P(11c|101) \\\hline 	    & \textbf{Q4} & &&& \textbf{Q8} &&&  \\ 	  r13 &&  \sum_a P(a00|011) & = & \sum_a P(a00|111) && \sum_c P(00c|110) & = & \sum_c P(00c|111) \\ 	  r14 &&  \sum_a P(a01|011) & = & \sum_a P(a01|111) && \sum_c P(01c|110) & = & \sum_c P(01c|111) \\ 	  r15 &&  \sum_a P(a10|011) & = & \sum_a P(a10|111) && \sum_c P(10c|110) & = & \sum_c P(10c|111) \\ 	  r16 &&  \sum_a P(a11|011) & = & \sum_a P(a11|111) && \sum_c P(11c|110) & = & \sum_c P(11c|111) \\\hline 	\end{array} 	\nonumber   \end{equation}
\begin{equation} 	\begin{array}{c|cccc|cccc} 	\mathcal{RC} && c1 & & c2 && c3 & & c4 \\\hline 	 & \textbf{Q1} & &&& \textbf{Q5} &&& \\ 	 r1 &&  \sum_{a,b} P(ab0|000) & = & \sum_{a,b} P(ab0|010) && \sum_{b,c} P(0bc|000) & = & \sum_{b,c} P(0bc|001) \\ 	 r2 && 				 & = & \sum_{a,b} P(ab0|100) && 			 & = & \sum_{b,c} P(0bc|010) \\ 	 r3 &&  			 & = & \sum_{a,b} P(ab0|110) && 			 & = & \sum_{b,c} P(0bc|011) \\ 	 & \textbf{Q2} & &&& \textbf{Q6} &&& \\ 	 r4 &&  \sum_{a,b} P(ab1|000) & = & \sum_{a,b} P(ab1|010) && \sum_{b,c} P(1bc|000) & = & \sum_{b,c} P(1bc|001) \\ 	 r5 && 				 & = & \sum_{a,b} P(ab1|100) && 			 & = & \sum_{b,c} P(1bc|010) \\ 	 r6 &&  			 & = & \sum_{a,b} P(ab1|110) && 			 & = & \sum_{b,c} P(1bc|011) \\ 	 & \textbf{Q3} & &&& \textbf{Q7} &&& \\ 	 r7 &&  \sum_{a,b} P(ab0|001) & = & \sum_{a,b} P(ab0|011) && \sum_{b,c} P(0bc|100) & = & \sum_{b,c} P(0bc|101) \\ 	 r8 && 				 & = & \sum_{a,b} P(ab0|101) && 			 & = & \sum_{b,c} P(0bc|110) \\ 	 r9 &&  			 & = & \sum_{a,b} P(ab0|111) && 			 & = & \sum_{b,c} P(0bc|111) \\ 	 & \textbf{Q4} & &&& \textbf{Q8} &&& \\ 	 r10 &&  \sum_{a,b} P(ab1|001) & = & \sum_{a,b} P(ab1|011) && \sum_{b,c} P(1bc|100) & = & \sum_{b,c} P(1bc|101) \\ 	 r11 && 			& = & \sum_{a,b} P(ab1|101) && 			 & = & \sum_{b,c} P(1bc|110) \\ 	 r12 &&  			& = & \sum_{a,b} P(ab1|111) && 			 & = & \sum_{b,c} P(1bc|111) \\ 	\end{array} 	\nonumber   \end{equation}

We use this table as a means to refer to its elements (terms of sums
of probabilities) using rows and columns  (e.g. $\sum_{a,b} P(ab0|010) \equiv \mathcal{RC}(1,2)$)
and to define sub-tables referred as sectors (e.g. $\mathcal{P}({\textbf{Q1}})$ or $\mathcal{RC}({\textbf{Q2}})$).

Consider  $\mathcal{P}$. In each sector $\mathcal{P}(\textbf{Qi})$, $i\in\{1,\ldots,8\}$, the last equality is implied by the previous ones and one of 8 normalization conditions in $\mathcal{N}$, which gives 8 dependent equalities. There are two more redundant conditions that can be found by writing two  sequences of equalities that begin and end with the same sum of probabilities, but with different rows or columns in the tables above. In sectors $\{\mathcal{P}(\textbf{Q1}),\mathcal{P}(\textbf{Q2}),\mathcal{P}(\textbf{Q5}),\mathcal{P}(\textbf{Q7})\}$ and $\{\mathcal{P}(\textbf{Q3}),\mathcal{P}(\textbf{Q4}),\mathcal{P}(\textbf{Q6}),\mathcal{P}(\textbf{Q8})\}$ we identify the corresponding two sequences (\ref{eq:D.1.1}) and (\ref{eq:D.1.2}) respectively. We designate these  kind of sequences as \textit{closed paths}. \begin{equation} 	\begin{array}{ccccccc} 	 \mathcal{P}(1,2) + \mathcal{P}(2,2) & = & \mathcal{P}(1,1) + \mathcal{P}(2,1) & = & \mathcal{P}(1,3) + \mathcal{P}(3,3) & = & \mathcal{P}(1,4) + \mathcal{P}(3,4)\\ 					     & = & \mathcal{P}(5,1) + \mathcal{P}(6,1) & = & \mathcal{P}(5,2) + \mathcal{P}(6,2) & = & \mathcal{P}(9,4) + \mathcal{P}(11,4)\\ 					     & = & \mathcal{P}(9,3) + \mathcal{P}(11,3)  	\end{array}   \label{eq:D.1.1}   \end{equation}  
\begin{equation} 	\begin{array}{ccccccc} 	 \mathcal{P}(9,2) + \mathcal{P}(10,2) & = & \mathcal{P}(9,1) + \mathcal{P}(10,1) & = & \mathcal{P}(5,3) + \mathcal{P}(7,3) & = & \mathcal{P}(5,4) + \mathcal{P}(7,4)\\ 					     & = & \mathcal{P}(13,1) + \mathcal{P}(14,1) & = & \mathcal{P}(13,2) + \mathcal{P}(14,2) & = & \mathcal{P}(13,4) + \mathcal{P}(15,4)\\ 					     & = & \mathcal{P}(13,3) + \mathcal{P}(15,3)  	\end{array}   \label{eq:D.1.2}   \end{equation}
\begin{equation} 	\begin{array}{c cc cc cc cc} 	  \sum_{ac}P(a0c|100) & = & \sum_aP(a00|100) + \sum_aP(a01|100) & = & \mathcal{P}(1,2) + \mathcal{P}(2,2)  	  \\ & =& \ldots  & = & \mathcal{P}(9,3) + \mathcal{P}(11,3) \\ 	  & = & \sum_cP(00c|100) + \sum_cP(10c|100) & = & \sum_{ac}P(a0c|100) 	\end{array}   \label{eq:D.1.3}   \end{equation}
\begin{equation} 	\begin{array}{c cc cc cc cc} 	  \sum_{ac}P(a0c|110) & = & \sum_aP(a00|110) + \sum_aP(a01|110) & = & \mathcal{P}(9,2) + \mathcal{P}(10,2)   	  \\ & = & \ldots  & = & \mathcal{P}(13,3) + \mathcal{P}(15,3)  \\ 	  & = & \sum_cP(00c|110) + \sum_cP(10c|110) & = & \sum_{ac}P(a0c|110) 	\end{array}   \label{eq:D.1.4}   \end{equation}
Notice that first and last terms in each pair $((\ref{eq:D.1.1}) , (\ref{eq:D.1.2}))$ and  $((\ref{eq:D.1.3}), (\ref{eq:D.1.4}))$, describe the same values.

From this observation, it  follows that one equality is dependent in $\{\mathcal{P}(\textbf{Q1}),\mathcal{P}(\textbf{Q2}),\mathcal{P}(\textbf{Q5}),\mathcal{P}(\textbf{Q7})\}$ and similarly one  in  $\{\mathcal{P}(\textbf{Q3}),\mathcal{P}(\textbf{Q4}),\mathcal{P}(\textbf{Q6}),\mathcal{P}(\textbf{Q8})\}$. This, for the first case, can be schematically represented as:
\begin{eqnarray}  &\left(     \begin{array}{ccc} 	  \left\{\begin{array}{c cc } 		\mathcal{P}(1,1) & = & \mathcal{P}(1,2) \\ 		\mathcal{P}(2,1) & = & \mathcal{P}(2,2) \\ 		\mathcal{P}(9,3) & = & \mathcal{P}(9,4) \\ 		\mathcal{P}(11,3) & = & \mathcal{P}(11,4) 	  \end{array}\right\} 	  & + & 	  \left\{\begin{array}{c cc } 		\mathcal{P}(1,2) + \mathcal{P}(2,2) & = &  \mathcal{P}(9,3) + \mathcal{P}(11,3)  	  \end{array}\right\} 	       \end{array}\right) &  \nonumber\\ 	  && \nonumber\\ 	  & \Downarrow &   \nonumber\\ 	  &&  \nonumber\\       &\left(       \begin{array}{ccc} 	  \left\{\begin{array}{c cc } 		\mathcal{P}(1,2) + \mathcal{P}(2,2) & = &  \mathcal{P}(9,3) + \mathcal{P}(11,3)  \\ 		\mathcal{P}(2,1) & = & \mathcal{P}(2,2) \\ 		\mathcal{P}(9,3) & = & \mathcal{P}(9,4) \\ 		\mathcal{P}(11,3) & = & \mathcal{P}(11,4) 	  \end{array}\right\} 	  & \Rightarrow & 	  \left\{\begin{array}{c cc } 		\mathcal{P}(1,1) & = & \mathcal{P}(1,2)  	  \end{array}\right\}      \end{array}\right)&  \nonumber   \end{eqnarray}For the second case an analogous reasoning shows the redundancy of one equation. Closed paths (\ref{eq:D.1.1}) and (\ref{eq:D.1.2}) are the shortest possible paths in $\mathcal{P}$ so there are no more dependent equalities leaving in total $8+22$ independent conditions for the set of constraints $\mathcal{N}\cup\mathcal{P}$. 

Now, consider the full set of RC constraints $\mathcal{N}\cup\mathcal{P}\cup\mathcal{RC}$. Due to the normalization conditions, it follows that each  sector $\mathcal{RC}(\textbf{Q}i+1)$, $i\in{1,3,5,7}$ is implied by $\mathcal{RC}(\textbf{Q}i)$ giving 12 dependent conditions. Furthermore in each of the remaining sectors of $\mathcal{RC}$ two out of three equalities are implied by $\mathcal{P}$. As an example consider sector $\mathcal{RC}(\textbf{Q1})$, then write:  \begin{eqnarray} 	\mathcal{RC}(1,1) = \mathcal{RC}(2,2) & \Leftrightarrow &  \mathcal{P}(1,1) + \mathcal{P}(3,1) = \mathcal{P}(1,2) + \mathcal{P}(3,2)\\ 	\mathcal{RC}(1,2) = \mathcal{RC}(3,2) & \Leftrightarrow &  \mathcal{P}(9,1) + \mathcal{P}(11,1) = \mathcal{P}(9,2) + \mathcal{P}(11,2)   \end{eqnarray}In other words two out of three equalities is sector  $\mathcal{RC}(\textbf{Q1})$ are implied by sectors $\mathcal{P}(\textbf{Q1})$ and $\mathcal{P}(\textbf{Q3})$. Analogously sectors $\{\mathcal{P}(\textbf{Q2},\mathcal{P}(\textbf{Q4}\}$, $\{\mathcal{P}(\textbf{Q5},\mathcal{P}(\textbf{Q6}\}$ and $\{\mathcal{P}(\textbf{Q7},\mathcal{P}(\textbf{Q8}\}$ leave only one independent equation in sectors   $\mathcal{RC}(\textbf{Q3})$, $\mathcal{RC}(\textbf{Q5})$ and $\mathcal{RC}(\textbf{Q7})$ respectively. In summary, the RC (3,2,2) polytope is fully described by 34 independent conditions so its dimensionality is $\mathcal{D}\left[RC\left(3,2,2\right)\right]=64-34=30$.

\subsection*{Step 3: Computing $\mathcal{D}\left[RC\left(3,m,n\right)\right]$}
We now proceed to compute the dimensionality of the RC polytope in the general $(3,m,n)$ scenario. 
Like in Step 2, we first consider the set $\mathcal{P}$. Notice that using normalization conditions we can delete   $2(m-1)$ equations in each of the $2 m^2$ sectors $\mathcal{P}(\textbf{Q})$. To construct closed paths between sectors one needs probabilities that for a given input and output of Bob, sum over all outputs of Alice and Charlie. This, due to normalization that removes e.g. last row in each sector, can be done uniquely for  $n-1$ outputs and  $m$ inputs of Bob  for any choice of $(m-1)^2$ combinations of columns  for Alice and  Charlie. This, in total, gives  $2m^2(m-1) + (n-1)m(m-1)^2$) dependent equalities and by Eq.(\ref{eq:1.8}), $2m^2n^2(m-1)+m^2n(2-m)+m(1-n)$ independent equalities. 

For the set $\mathcal{N}\cup\mathcal{P}\cup\mathcal{RC}$ normalization conditions together with sectors $\{\mathcal{RC}{(\textbf{Qi})},\ldots,\mathcal{RC}{(\textbf{Qi+m-1})}\}$ imply sector $\mathcal{RC}{(\textbf{Qi+m})}$ for $i\in\{l\cdot m\}$ with $l=\{0,1,\ldots,2*(m-1)\}$ leaving  $2(n-1)m$ sectors. By a similar argument as in the (3,2,2) scenario, in each remaining sector $\mathcal{RC}(\textbf{Q})$ constraints in $\mathcal{P}$ imply all sums of probabilities with the same input of Bob leaving only   $m-1$ equations. This gives $2(n-1)m(m-1)$ independent equalities. Subtracting the total number of independent conditions from $(m\cdot n)^3$ gives the dimensionality of RC polytope in $(3,m,n)$ scenario as:   
\begin{equation}
\mathcal{D}\left[RC\left(3,m,n\right)\right]=[m(n\text{\textminus}1)+1]^{3}+m^{2}(m\text{\textminus}1)(n\text{\textminus}1)^{2}\text{\textminus}1
\end{equation}

\section*{Appendix F: Nontrivial bounds for Relativistic Causal correlations.}



The main contribution of our article is the proof that  two eavesdroppers can collaborate to break any device-independent security protocol if they can prepare devices with the strongest correlations allowed by RC theories. The natural assumption that the eavesdroppers  can choose freely space-time positions is shown here to be relevant for the proof of the No-go theorem. The reason is because  eavesdroppers must choose necessarily appropriate positions for the measuring devices to reach the strongest correlations allowed by RC. In this appendix we show how a restriction in the space-time positions of the eavesdroppers could limit the devices correlations even bellow the strength of quantum correlations, demonstrating that the selection of the space-time positions is crucial for the attack of the eavesdroppers. 
\\

We firstly study trade-off relations between three-party Svetlichny expressions $\langle \mathcal{I} \rangle_{ACD}, \langle \mathcal{I} \rangle_{BCD}$ of the form
\begin{equation}
\left\langle \mathcal{I}\right\rangle _{ACD}+\left\langle \mathcal{I}\right\rangle _{BCD}\leq2\mathcal{B}
\end{equation}
where $\mathcal{B}$ is the so-called ''Broadcast" bound. 
We remark that the distinguishing
feature of RC correlations is the point to region (PTR) signaling,
described in detail in Appendix A, namely that in certain measurement configurations, a single party can signal to a region thus influencing the correlations between two or more other parties.

Consider a three-party situation with measurement inputs $x,y,z$ and outputs $a,b,c$ for Alice, Bob and Charlie respectively. Broadcasting correlations represent the situation when  one party sends all the information about its measurement setting and outcome to the other two parties. In \cite{ref29, sp0}, it was pointed out that quantum correlations violate broadcasting correlations and this can be regarded as an alternative notion of genuine multi-partite nonlocality. Tripartite broadcasting correlations $P(a,b,c|x,y,z)$ are defined as follows, 
\begin{equation}\label{BC3} \begin{split} & P(a,b,c|x,y,z)\\& \quad =  \sum_{\lambda_1} q(\lambda_1)  P(a|x,\lambda_1) P(b|y,x,a,\lambda_1) P(c|z,x,a,\lambda_1) \\& \quad + \sum_{\lambda_2} q(\lambda_2) P(b|y,\lambda_2) P(a|x,y,b,\lambda_2) P(c|z,y,b,\lambda_2) \\& \quad +\sum_{\lambda_3} q(\lambda_3) P(c|z,\lambda_3) P(b|y,z,c,\lambda_3) P(a|x,z,c,\lambda_3). \end{split} 
\end{equation}  
Observe that in the first term, Bob's output $b$ and Charlie's output $c$ depend upon Alice's input and output $x,a$, in the case where Alice has broadcast these, and similarly for the other two terms. The following lemma makes a connection between broadcast correlations (BC) and relativistic causal (RC) correlations, under the constraint that some of the observables are jointly measurable. 

\begin{lemma}[2]
Any RC tripartite probability distribution can be realized by a broadcast model with the additional condition that all the observables, measured by one party who does not signal PTR, are co-measurable
\end{lemma}


\begin{proof}
	Like in Section II of the main text, we consider the tripartite spacetime measurement configuration in Fig. \ref{Figure1} 
	where Bob signals PTR (i.e. to the correlations between $A$ and $C$) so that the RC
	constraints are given by the set of equations
	\begin{eqnarray}
	\sum_{a}P\left(a,b,c\mid x,y,z\right) & = & \sum_{a}P\left(a,b,c\mid x^{\prime},y,z\right)\:\forall x,x^{\prime},y,z,b,c\label{eq:AtoBC}\\
	\sum_{c}P\left(a,b,c\mid x,y,z\right) & = & \sum_{c}P\left(a,b,c\mid x,y,z^{\prime}\right)\:\forall z,z^{\prime},x,y,a,b\label{eq:CtoAB}\\
	\sum_{b,c}P\left(a,b,c\mid x,y,z\right) & = & \sum_{b,c}P\left(a,b,c\mid x,y^{\prime},z^{\prime}\right)\:\forall y,y^{\prime},z,z^{\prime},x,a\label{eq:BCtoA}\\
	\sum_{a,b}P\left(a,b,c\mid x,y,z\right) & = & \sum_{a,b}P\left(a,b,c\mid x^{\prime},y^{\prime},z\right)\:\forall x,x^{\prime},y,y^{\prime},z,c\label{eq:BAtoC}
	\end{eqnarray}
	From the first two conditions, we also clearly have,
	\begin{equation} \sum_{a,c}P(a,b,c|x,y,z) = \sum_{a,c} P(a,b,c|x',y,z').  \end{equation}
	This implies that $P(b|x,y,z)=P(b|y)$ is independent of $x,z$. Now, any RC tripartite probability distribution can be written as, 
	\begin{eqnarray}  P(a,b,c|x,y,z) &=& P(a,c|x,y,z,b) P(b|x,y,z) \nonumber \\ &=& P(a,c|x,y,z,b) P(b|y) 
	\end{eqnarray}
	Without loss of generality let's say that all the observables $x$ measured by Alice are co-measurable. Also, let's remember some useful concepts:  a \emph{commutation graph} is a graph with vertices representing observables, edges connecting observables that are jointly measurable and a \emph{chordal graph}  is a graph in which all cycles of four or more vertices have a chord going through them.

 In our case we can define a commutation graph of all the observables measured by Alice and Charlie conditioned on a particular pair of Bob's observable and outcome $y,b$. In the commutation graph, all pairs $x,x'$ and $x,z$ are connected, so that this commutation graph is chordal.  For chordal graphs of measurements corresponds an expression for which a joint probability distribution exists and which is hence classical \cite{RS12}. Therefore  exists an overall joint probability distribution of all $x,z$ conditioned on $y,b$. By the Fine's theorem \cite{fine} we conclude that $P(a,c|x,y,z,b) = P(a|x,y,b)P(c|y,z,b)$. Thus, \begin{equation} P(a,b,c|x,y,z) = P(a|x,y,b) P(c|y,z,b)P(b|y)\end{equation}
which is a particular form of the broadcast correlations given in  \eqref{BC3} in which $q(\lambda_1)=q(\lambda_3)=0$ and $\lambda_2$ is unique. 
\end{proof}

Secondly, we consider the Bell scenario involving four spatially separated parties Alice(A), Bob(B), Charlie(C) and Dave(D). Consider any broadcasting inequality $\left\langle \mathcal{I}\right\rangle_{ACD}$ between Alice, Charlie and Dave in which Alice has two measurement settings $x=0,1$. Assume now that $x=0,1$ are co-measurable, then by Lemma 2 :  \begin{equation} \label{s} \left\langle \mathcal{I}\right\rangle_{ACD} = \left\langle \mathcal{I}\right\rangle^{a_0}_{ACD} + \left\langle \mathcal{I}\right\rangle^{a_1}_{ACD} \leq \mathcal{B}, \end{equation} where $\mathcal{B}$ is the upper bound on broadcasting correlations \eqref{BC3}, and $\left\langle \mathcal{I}\right\rangle^{a_0}_{ACD}, \left\langle \mathcal{I}\right\rangle^{a_1}_{ACD}$ are the expressions corresponding to $x=0, 1$ respectively.

\begin{figure}
	
	\includegraphics[scale=0.8]{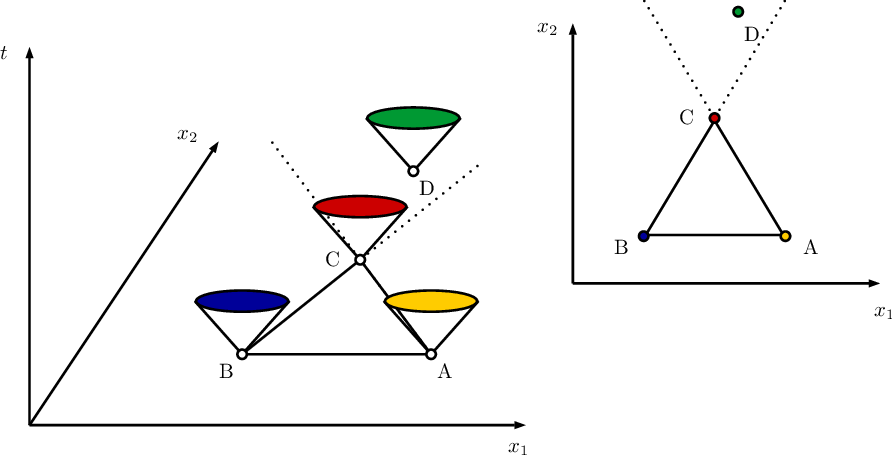}

	\caption{\label{Figure4} Two perspectives of four parties A,B,C,D in a (2+1 D) dimensional spacetime. The four parties make a simultaneous measurement in the particular reference frame of the picture. The measurement events of parties A,B,C form a triangle and party D is in some location inside the region defined by line BC and line AC (Dave's region). The correlations then satisfy a tight Monogamy relation for any broadcast inequality, for instance Svetlichny's inequality \cite{ref29}: $\left\langle \mathcal{I}_{Sve}\right\rangle _{ACD}+\left\langle \mathcal{I}_{Sve}\right\rangle _{BCD}\overset{RC}{\leq}8$. }
\end{figure}

\begin{prop}[3]
 In the four party scenario if the following two conditions hold, \\ (1) $A$ and $B$ do not signal PTR, \\ (2) any observable measured by $A$ and any observables measured by $B$ are non-disturbing (or alternatively no party signals PTR such that it affects the correlations between $A$ and $B$), \\ then the monogamy relation, \begin{equation}\label{mr} \left\langle \mathcal{I}\right\rangle_{ACD}+\left\langle \mathcal{I}\right\rangle_{BCD}\leq 2\mathcal{B}\end{equation} is satisfied in all theories obeying relativistic causality.

\end{prop}


\begin{proof} The expression of interest can be written as, \begin{equation}\label{Bbound} \left\langle \mathcal{I}\right\rangle_{ACD}+\left\langle \mathcal{I}\right\rangle_{BCD} = (\left\langle \mathcal{I}\right\rangle^{a_0}_{ACD} + \left\langle \mathcal{I}\right\rangle^{b_1}_{BCD}) + (\left\langle \mathcal{I}\right\rangle^{b_0}_{BCD} +\left\langle \mathcal{I}\right\rangle^{a_1}_{ACD} ) \end{equation} The terms within each bracket can be interpreted as the same inequality $\mathcal{I}$ in which the first party measures $x=0, y=1$ and the second measures $x=1,y=0$. Now, any two observables measured by Alice and Bob are non-disturbing and jointly measurable since no other party signals PTR to influence the correlations between them. Moreover, both the parties do not signal PTR to affect the correlation of others. Thus, from the above Lemma 2, one concludes that each of the two terms is bounded by its broadcasting value within theories obeying relativistic causality, that is, $\mathcal{B}$. Hence, the whole expression is bounded by $2\mathcal{B}$. 
\end{proof}
An example of a measurement configuration given by the space-time location of four parties' measurement events is shown in Figure 4 where the two conditions given in Proposition 3 hold. This example shows that if eavesdroppers are constrained to space-time positions like those allowed to Dave, their correlations are bounded by BC, which are known to be weaker than quantum correlations  \cite{sp0}. This limitation introduced by the restriction on space-time positions --to Dave's region for instance-- seat aside the attack of eavesdroppers since the reliable parties (Alice, Bob and Charlie in the example)  could perform an experiment with quantum correlations  they could not reproduce.

\newpage

\section*{Appendix G: List of Extremal boxes}
\renewcommand*{\arraystretch}{1.4}
\begin{longtable}{ccl}
	\hline
	Class & Prob. & Condition for RC Extremal Boxes Beyond No-signaling Polytope  \\
	\hline
	
	1
	&$1 $&$ a b c (1 \oplus x) (1 \oplus z) == 1$\\	
	&$\frac{1}{2} $&$ b (c x \oplus (a \oplus x y) z) == 1$\\	
	\hline 
	2
	&$1 $&$ a b c (1 \oplus x) y (1 \oplus z) == 1$\\	
	&$\frac{1}{2} $&$ a (c \oplus c y \oplus b z) \oplus b x (c \oplus z \oplus y z) == 1$\\	
	\hline 
	3
	&$\frac{1}{4} $&$ (1 \oplus c) x y \oplus b (c \oplus y \oplus y z \oplus x y z) \oplus a (c \oplus y \oplus z \oplus b z \oplus y z \oplus b c x y z) == 1$\\	
	&$\frac{3}{4} $&$ a b c x y z == 1$\\	
	&$\frac{1}{2} $&$ a b (1 \oplus c \oplus y \oplus z \oplus y z \oplus c x y z) == 1$\\	
	\hline 
	4
	&$\frac{1}{3} $&$ a y (c \oplus z) \oplus b (1 \oplus a \oplus a c \oplus c x \oplus x y \oplus y z \oplus x y z) == 1$\\	
	&$\frac{2}{3} $&$ b c ((1 \oplus x) y z \oplus a (y z \oplus x (1 \oplus y \oplus y z))) == 1$\\	
	\hline 
	5
	&$\frac{1}{5} $&$ x \oplus c x \oplus y \oplus c y \oplus y z \oplus x y z \oplus b (c \oplus x \oplus c x \oplus y \oplus z \oplus x y z) \oplus a (1 \oplus c \oplus y z \oplus b (1 \oplus c y \oplus z)) == 1$\\	
	&$\frac{4}{5} $&$ a b c (1 \oplus x) y (1 \oplus z) == 1$\\	
	&$\frac{2}{5} $&$ b y (c x \oplus a z) == 1$\\	
	&$\frac{3}{5} $&$ a b c (1 \oplus y) == 1$\\	
	\hline 
	6
	&$\frac{1}{3} $&$ (1 \oplus c) x y \oplus a (1 \oplus c \oplus b c \oplus b z) \oplus b (c \oplus y \oplus y z \oplus x y z) == 1$\\	
	&$\frac{2}{3} $&$ a b c x y z == 1$\\	
	\hline 
	7
	&$\frac{1}{4} $&$ x (c \oplus y) \oplus b (c \oplus y \oplus z \oplus x z) \oplus a (c \oplus y \oplus b z \oplus b c x y z) == 1$\\	
	&$\frac{3}{4} $&$ a b c x y z == 1$\\	
	&$\frac{1}{2} $&$ a b (1 \oplus c \oplus y \oplus z \oplus y z \oplus c x y z) == 1$\\	
	\hline 
	8
	&$\frac{1}{4} $&$ c x \oplus y \oplus c y \oplus x y \oplus x z \oplus x y z \oplus b (c \oplus y \oplus x z) \oplus a (y \oplus z \oplus b z \oplus y z \oplus c (1 \oplus b (x \oplus y) z)) == 1$\\	
	&$\frac{3}{4} $&$ a b c (x \oplus y) z == 1$\\	
	&$\frac{1}{2} $&$ a b (1 \oplus c \oplus y \oplus z \oplus c x z \oplus y z \oplus c y z) == 1$\\	
	\hline 
	9
	&$\frac{1}{3} $&$ x y (c \oplus z) \oplus b (c \oplus x \oplus x y \oplus z \oplus y z \oplus x y z) \oplus a (b (1 \oplus c) \oplus y (c \oplus z)) == 1$\\	
	&$\frac{2}{3} $&$ a b c (x \oplus x y \oplus z \oplus y z \oplus x y z) == 1$\\	
	\hline 
	10
	&$\frac{1}{3} $&$ x (c \oplus y \oplus z) \oplus a (b \oplus c \oplus b c \oplus y \oplus z) \oplus b (c \oplus y z \oplus x (y \oplus z \oplus y z)) == 1$\\	
	&$\frac{2}{3} $&$ a b c x (y \oplus z) == 1$\\	
	\hline 
	11
	&$\frac{1}{3} $&$ c x y \oplus a (y \oplus z) \oplus b (1 \oplus a \oplus c \oplus a c \oplus y \oplus z \oplus x y z) == 1$\\	
	&$\frac{2}{3} $&$ a (1 \oplus b) c (1 \oplus y \oplus z \oplus x y z) == 1$\\	
	\hline 
	12
	&$\frac{1}{4} $&$ (a \oplus x) y (c \oplus z) \oplus b (1 \oplus a \oplus c x \oplus y \oplus x z \oplus a c x y z) == 1$\\	
	&$\frac{3}{4} $&$ a b c x y z == 1$\\	
	&$\frac{1}{2} $&$ b c (1 \oplus x) y \oplus a c (1 \oplus b \oplus y \oplus b x y z) == 1$\\	
	\hline 
	13
	&$\frac{1}{3} $&$ c x (1 \oplus y) \oplus a (c \oplus y z) \oplus b (1 \oplus a (1 \oplus c) \oplus y z \oplus x (1 \oplus y) (1 \oplus z)) == 1$\\	
	&$\frac{2}{3} $&$ (1 \oplus a) b c y z == 1$\\	
	\hline 
	14
	&$\frac{1}{3} $&$ x y (c \oplus z) \oplus a (b \oplus c \oplus b c \oplus y z) \oplus b (c \oplus x \oplus x y \oplus z \oplus y z \oplus x y z) == 1$\\	
	&$\frac{2}{3} $&$ a b c x y z == 1$\\	
	\hline 
	15
	&$\frac{1}{3} $&$ y \oplus c y \oplus b (c \oplus y \oplus x z \oplus x y z) \oplus a (1 \oplus c \oplus b (1 \oplus c \oplus z)) == 1$\\	
	&$\frac{2}{3} $&$ a b c y (1 \oplus z) == 1$\\	
	\hline 
	16
	&$\frac{1}{4} $&$ (a \oplus x) y (c \oplus z) \oplus b (1 \oplus a \oplus y \oplus a c x y z \oplus x (1 \oplus c \oplus y \oplus y z)) == 1$\\	
	&$\frac{3}{4} $&$ a b c x y z == 1$\\	
	&$\frac{1}{2} $&$ b c (1 \oplus x) y \oplus a c (1 \oplus b \oplus y \oplus b x y z) == 1$\\	
	\hline 
	17
	&$\frac{1}{5} $&$ b (c \oplus x \oplus c x \oplus y \oplus z \oplus x y z \oplus a (1 \oplus c y \oplus z)) == 1$\\	
	&$\frac{4}{5} $&$ a (1 \oplus b) c (1 \oplus x) y (1 \oplus z) == 1$\\	
	&$\frac{2}{5} $&$ a y z \oplus b (a y \oplus (1 \oplus a) x (1 \oplus y)) z \oplus c x ((1 \oplus a) b z \oplus y (1 \oplus b \oplus b z \oplus a b z)) == 1$\\	
	&$\frac{3}{5} $&$ a (1 \oplus b) c (1 \oplus y) == 1$\\	
	\hline 
	18
	&$\frac{1}{3} $&$ (1 \oplus c) x y \oplus a (1 \oplus c \oplus b c \oplus z \oplus b z \oplus y z) \oplus b (c \oplus y \oplus y z \oplus x y z) == 1$\\	
	&$\frac{2}{3} $&$ a b c x y z == 1$\\	
	\hline 
	19
	&$\frac{3}{5} $&$ (1 \oplus a) b c (1 \oplus y) == 1$\\	
	&$\frac{2}{5} $&$ (1 \oplus b) y (1 \oplus c \oplus z) \oplus a (1 \oplus c \oplus y z \oplus b (1 \oplus c \oplus c x y \oplus y z \oplus c x y z)) == 1$\\	
	&$\frac{1}{5} $&$ y ((b \oplus x) (c \oplus z) \oplus a (c \oplus b c \oplus z)) == 1$\\	
	\hline 
	20
	&$\frac{1}{2} $&$ (b c \oplus a (1 \oplus b \oplus c)) (1 \oplus y) == 1$\\	
	&$\frac{1}{4} $&$ y (a \oplus b \oplus x (c \oplus z)) == 1$\\	
	\hline 
	21
	&$\frac{1}{2} $&$ (a c \oplus b (1 \oplus a \oplus c)) (1 \oplus y) == 1$\\	
	&$\frac{1}{4} $&$ y (a \oplus b \oplus c x \oplus b x z) == 1$\\	
	\hline 
	22
	&$\frac{1}{2} $&$ b (c \oplus c y \oplus x y z \oplus c x y z) \oplus a (1 \oplus c \oplus y \oplus c y \oplus c x y z \oplus b (1 \oplus y \oplus x y z)) == 1$\\	
	&$\frac{1}{4} $&$ y (a \oplus b \oplus c \oplus c x \oplus a z \oplus b x z) == 1$\\	
	\hline 
	23
	&$\frac{1}{2} $&$ b (1 \oplus c) x y \oplus a (c x y (1 \oplus z) \oplus b (1 \oplus c \oplus c x y \oplus c x y z)) == 1$\\	
	&$\frac{1}{4} $&$ c (x \oplus y) \oplus b (c \oplus c x \oplus z \oplus y z) \oplus a (c (1 \oplus y) \oplus (b \oplus y) z) == 1$\\	
	\hline 
	24
	&$\frac{1}{2} $&$ (1 \oplus b) (1 \oplus c) x y \oplus a (1 \oplus c \oplus b (1 \oplus c \oplus c x y \oplus c x y z)) == 1$\\	
	&$\frac{1}{4} $&$ c x (1 \oplus y) \oplus a (c \oplus c y \oplus b z) \oplus b (c \oplus c x \oplus z \oplus y z) == 1$\\	
	\hline 
	25
	&$\frac{1}{2} $&$ x y \oplus c x y \oplus b x (c \oplus y \oplus c z \oplus c y z) \oplus a (1 \oplus c \oplus b (1 \oplus c \oplus c x \oplus c x z)) == 1$\\	
	&$\frac{1}{4} $&$ b (c (1 \oplus x) \oplus (a \oplus x \oplus y) z) == 1$\\	
	\hline 
	26
	&$\frac{1}{4} $&$ c x y \oplus b (1 \oplus c \oplus y \oplus a (1 \oplus c \oplus c y)) \oplus a y (1 \oplus z) == 1$\\	
	&$\frac{3}{4} $&$ a (1 \oplus b) c (1 \oplus y) == 1$\\	
	&$\frac{1}{2} $&$ a c y ((1 \oplus x) z \oplus b (x \oplus z)) == 1$\\	
	\hline 
	27
	&$\frac{1}{3} $&$ x y (c \oplus z) \oplus a (b \oplus c \oplus b c \oplus z) \oplus b (c \oplus y (1 \oplus x \oplus z)) == 1$\\	
	&$\frac{2}{3} $&$ a b c x y z == 1$\\	
	\hline 
	28
	&$\frac{1}{4} $&$ y \oplus c y \oplus a (1 \oplus c \oplus b z \oplus b c x y z) \oplus b (y z \oplus x (1 \oplus c \oplus z)) == 1$\\	
	&$\frac{3}{4} $&$ a b c x y z == 1$\\	
	&$\frac{1}{2} $&$ b c (1 \oplus a \oplus x \oplus y \oplus x y \oplus a x y z) == 1$\\	
	\hline 
	29
	&$\frac{1}{3} $&$ a (1 \oplus b \oplus c \oplus b c \oplus y) \oplus y (b \oplus c x \oplus b x z) == 1$\\	
	&$\frac{2}{3} $&$ (1 \oplus a) b c (1 \oplus y) == 1$\\	
	\hline 
	30
	&$\frac{1}{3} $&$ x y (c \oplus z) \oplus a (b \oplus c \oplus b c \oplus z) \oplus b (c \oplus x \oplus y \oplus x z \oplus y z \oplus x y z) == 1$\\	
	&$\frac{2}{3} $&$ a b c x y z == 1$\\	
	\hline 
	31
	&$\frac{1}{4} $&$ c \oplus b x \oplus b c x \oplus y \oplus b x y \oplus z \oplus b z \oplus a (1 \oplus c \oplus y \oplus z \oplus b z \oplus b c x y z) == 0$\\	
	&$\frac{3}{4} $&$ a b c x y z == 1$\\	
	&$\frac{1}{2} $&$ a (b y (1 \oplus z) \oplus c (1 \oplus b \oplus y \oplus b x y z)) == 1$\\	
	\hline 
	32
	&$\frac{1}{3} $&$ c x (1 \oplus y) \oplus b (c \oplus y) \oplus a (b \oplus c \oplus b c \oplus c y) == 1$\\	
	&$\frac{2}{3} $&$ a (1 \oplus b) c y == 1$\\	
	\hline 
	33
	&$\frac{1}{3} $&$ c x (1 \oplus y) \oplus a (c \oplus y z) \oplus b (1 \oplus a \oplus a c \oplus y z) == 1$\\	
	&$\frac{2}{3} $&$ (1 \oplus a) b c y z == 1$\\	
	\hline 
	34
	&$\frac{3}{5} $&$ a (1 \oplus b) y (c \oplus z) == 1$\\	
	&$\frac{2}{5} $&$ b x (c \oplus c y \oplus z \oplus y z \oplus c y z) \oplus a (b (1 \oplus y) z \oplus c (1 \oplus y \oplus b x y z)) == 1$\\	
	&$\frac{1}{5} $&$ b (1 \oplus c \oplus x \oplus c x \oplus y \oplus x y z \oplus a (1 \oplus c \oplus c y \oplus y z)) == 1$\\	
	\hline 
	35
	&$\frac{1}{2} $&$ b (1 \oplus x) (c \oplus y \oplus y z \oplus c y z) \oplus a (c y \oplus b (1 \oplus z \oplus c z \oplus x (1 \oplus c \oplus z \oplus c y z))) == 1$\\	
	&$\frac{1}{4} $&$ b ((a \oplus y) z \oplus x (1 \oplus c \oplus a c z \oplus a c y z)) == 1$\\	
	&$\frac{3}{4} $&$ a b c x (1 \oplus y) z == 1$\\	
	\hline 
	36
	&$\frac{1}{2} $&$ (1 \oplus a) c y \oplus b ((1 \oplus a \oplus y \oplus a x y) z \oplus c (1 \oplus a \oplus a y z \oplus a x y z)) == 1$\\	
	&$\frac{1}{4} $&$ b x y \oplus c x (1 \oplus b \oplus y) \oplus a (b \oplus c \oplus c y \oplus b z) == 1$\\	
	\hline 
	37
	&$\frac{1}{2} $&$ b (1 \oplus x) (y \oplus c (1 \oplus z \oplus y z)) \oplus a ((1 \oplus c) (1 \oplus y) \oplus b (1 \oplus x y z \oplus c (y \oplus z \oplus y z \oplus x (1 \oplus y \oplus z)))) == 1$\\	
	&$\frac{1}{4} $&$ c (a \oplus x) y \oplus b ((1 \oplus a \oplus y \oplus a x y) z \oplus c (x \oplus a x y z)) == 1$\\	
	&$\frac{3}{4} $&$ a b (1 \oplus c) x y z == 1$\\	
	\hline 
	38
	&$\frac{1}{2} $&$ b c x (y \oplus z) \oplus a ((1 \oplus b \oplus y \oplus b x y) z \oplus c (1 \oplus b \oplus y \oplus b x z)) == 1$\\	
	&$\frac{1}{4} $&$ a (b \oplus c y \oplus b z) \oplus b (c \oplus x \oplus c x \oplus y \oplus x z \oplus y z \oplus x y z) == 1$\\	
	\hline 
	39
	&$\frac{1}{2} $&$ b x (c \oplus y) z \oplus a ((1 \oplus b \oplus y) z \oplus c (1 \oplus b \oplus y \oplus b x y \oplus b x z)) == 1$\\	
	&$\frac{1}{4} $&$ a (b \oplus c y \oplus b z) \oplus b (c (1 \oplus x) \oplus (x \oplus y \oplus x y) (1 \oplus z)) == 1$\\	
	\hline 
	40
	&$\frac{1}{2} $&$ c (a \oplus x) y \oplus b ((a \oplus x \oplus a y \oplus x y \oplus a x y) z \oplus c (a (y \oplus z \oplus y z) \oplus x (1 \oplus a (1 \oplus y \oplus z)))) == 1$\\	
	&$\frac{1}{4} $&$ a (b \oplus c \oplus c y \oplus b z) \oplus b (1 \oplus x) (c \oplus y z) == 1$\\	
	\hline 
	41
	&$\frac{1}{2} $&$ b (1 \oplus x) (c \oplus y \oplus y z) \oplus a (c (1 \oplus x) y (1 \oplus z) \oplus b (1 \oplus z \oplus c (x \oplus z \oplus x y z))) == 1$\\	
	&$\frac{1}{4} $&$ c x y \oplus a y z \oplus b ((a \oplus x \oplus x y) z \oplus c x (1 \oplus a (1 \oplus y) z)) == 1$\\	
	&$\frac{3}{4} $&$ a b c x (1 \oplus y) z == 1$\\	
	\hline 
	42
	&$\frac{1}{2} $&$ b (1 \oplus c) x (1 \oplus y) z \oplus a (b (x \oplus y \oplus x y) z \oplus c (1 \oplus b \oplus y \oplus b x y \oplus b x z)) == 1$\\	
	&$\frac{1}{4} $&$ c x y \oplus a (b \oplus c y \oplus b z) \oplus b (c \oplus x \oplus c x \oplus y \oplus z \oplus x y z) == 1$\\	
	\hline 
	43
	&$\frac{1}{2} $&$ b (1 \oplus x) (c \oplus y \oplus y z \oplus c y z) \oplus a (c y \oplus b ((1 \oplus x y) (1 \oplus z) \oplus c (x \oplus z \oplus x y z))) == 1$\\	
	&$\frac{1}{4} $&$ b ((a \oplus y) z \oplus x (y \oplus z \oplus y z) \oplus c x (1 \oplus a (1 \oplus y) z)) == 1$\\	
	&$\frac{3}{4} $&$ a b c x (1 \oplus y) z == 1$\\	
	\hline 
	44
	&$\frac{1}{3} $&$ c x \oplus y \oplus c y \oplus x y \oplus x z \oplus y z \oplus a (b \oplus c \oplus b c \oplus y \oplus z) \oplus b (c \oplus x z \oplus y (1 \oplus x \oplus z)) == 1$\\	
	&$\frac{2}{3} $&$ a b c (x z \oplus y (1 \oplus x \oplus z)) == 1$\\	
	\hline 
	45
	&$\frac{1}{3} $&$ x y (c \oplus z) \oplus a (1 \oplus b \oplus c \oplus b c \oplus y \oplus z) \oplus b (c \oplus y (1 \oplus x \oplus z)) == 1$\\	
	&$\frac{2}{3} $&$ a b c x y z == 1$\\	
	\hline 
	46
	&$\frac{1}{3} $&$ y (c \oplus z) \oplus a (c \oplus y z) \oplus b (1 \oplus a \oplus a c \oplus c x \oplus y \oplus x z \oplus y z \oplus x y z) == 1$\\	
	&$\frac{2}{3} $&$ a b c x y z == 1$\\	
	\hline 
	47
	&$\frac{1}{2} $&$ b c x y z \oplus a (y (c \oplus z) \oplus b (y z \oplus c (1 \oplus z \oplus y z \oplus x (1 \oplus y \oplus z)))) == 1$\\	
	&$\frac{1}{4} $&$ c (1 \oplus b \oplus b x \oplus y) \oplus a (1 \oplus b \oplus y \oplus b z) \oplus x (b y \oplus z \oplus y z \oplus b y z) == 1$\\	
	\hline 
	48
	&$\frac{1}{2} $&$ b (1 \oplus c) x (1 \oplus y) z \oplus a (b x (1 \oplus y) z \oplus c (1 \oplus y z \oplus x y (1 \oplus z) \oplus b (1 \oplus y z \oplus x (y \oplus z)))) == 1$\\	
	&$\frac{1}{4} $&$ c x y \oplus a (b \oplus b z \oplus y z) \oplus b (c \oplus x \oplus c x \oplus y \oplus z \oplus x y z) == 1$\\	
	\hline 
	49
	&$\frac{1}{2} $&$ a c (x z \oplus y (1 \oplus x \oplus z)) \oplus b (a c x z \oplus y ((1 \oplus x) (1 \oplus c \oplus z) \oplus a (1 \oplus z \oplus c (x \oplus z)))) == 1$\\	
	&$\frac{1}{4} $&$ c (1 \oplus x \oplus b x \oplus y) \oplus b x (1 \oplus y) z \oplus a (1 \oplus y \oplus z \oplus b z) == 1$\\	
	\hline 
	50
	&$\frac{1}{2} $&$ b c (1 \oplus x) y z \oplus a (y (c \oplus z) \oplus b (y z \oplus c (x \oplus y \oplus x y \oplus x z \oplus y z))) == 1$\\	
	&$\frac{1}{4} $&$ b y \oplus b x y \oplus c (1 \oplus b x \oplus y) \oplus b z \oplus x z \oplus x y z \oplus b x y z \oplus a (1 \oplus b \oplus y \oplus b z) == 1$\\	
	\hline 
	51
	&$\frac{1}{2} $&$ b c x y \oplus a ((1 \oplus c) x y z \oplus b (x y z \oplus c (1 \oplus z \oplus y z \oplus x (1 \oplus y \oplus z)))) == 1$\\	
	&$\frac{1}{4} $&$ c (1 \oplus b \oplus b x \oplus x y) \oplus (1 \oplus b \oplus x \oplus b x y) z \oplus a (1 \oplus b \oplus b z \oplus y z) == 1$\\	
	\hline 
	52
	&$\frac{1}{2} $&$ b c x (1 \oplus y) z \oplus a ((1 \oplus b) (1 \oplus y) z \oplus c (1 \oplus y \oplus b (1 \oplus y z \oplus x (y \oplus z)))) == 1$\\	
	&$\frac{1}{4} $&$ y \oplus c y \oplus x y z \oplus a (b \oplus y \oplus b z) \oplus b (c (1 \oplus x) \oplus y \oplus x (1 \oplus y) (1 \oplus z)) == 1$\\	
	\hline 
	53
	&$\frac{1}{3} $&$ (1 \oplus c) x y \oplus a (1 \oplus c \oplus b (1 \oplus c \oplus z)) \oplus b (c \oplus y z \oplus x (y \oplus z)) == 1$\\	
	&$\frac{2}{3} $&$ a b c x y (1 \oplus z) == 1$\\	
	\hline 
	54
	&$\frac{2}{3} $&$ a c (1 \oplus y \oplus b (1 \oplus y \oplus y z \oplus x y z)) == 1$\\	
	&$\frac{1}{3} $&$ b x (1 \oplus c \oplus y) \oplus a (b (1 \oplus c) \oplus y (c \oplus z)) == 1$\\	
	\hline 
	55
	&$\frac{2}{3} $&$ a c (1 \oplus y \oplus b (1 \oplus y \oplus x y z)) == 1$\\	
	&$\frac{1}{3} $&$ (a \oplus x) y (c \oplus z) \oplus b (1 \oplus a \oplus c \oplus a c \oplus y \oplus x y z) == 1$\\	
	\hline 
	56
	&$\frac{1}{2} $&$ c (1 \oplus x) y \oplus a (b \oplus c y) \oplus b (c \oplus x y) == 1$\\	
	\hline 
	57
	&$\frac{1}{2} $&$ b x (c \oplus y) \oplus a (b \oplus c y) == 1$\\	
	\hline 
	58
	&$\frac{1}{2} $&$ b (c \oplus y z) \oplus a (1 \oplus b \oplus c \oplus y z) == 1$\\	
	\hline 
	59
	&$\frac{1}{2} $&$ b x (1 \oplus c \oplus y) \oplus a (b \oplus c \oplus y) == 1$\\	
	\hline 
	60
	&$\frac{1}{3} $&$ x y (c \oplus z) \oplus a (b \oplus c \oplus b c \oplus y z) \oplus b (c \oplus x y z) == 1$\\	
	&$\frac{2}{3} $&$ a b c x y z == 1$\\	
	\hline 
	61
	&$\frac{1}{3} $&$ a y (c \oplus z) \oplus b (1 \oplus a \oplus a c \oplus c x \oplus x y \oplus z \oplus x z) == 1$\\	
	&$\frac{2}{3} $&$ b c ((1 \oplus x) z \oplus a (z \oplus x (1 \oplus y \oplus z))) == 1$\\	
	\hline 
	62
	&$\frac{1}{3} $&$ x (1 \oplus y) (c \oplus z) \oplus a (b (1 \oplus c) \oplus (1 \oplus y) (c \oplus z)) \oplus b (c \oplus y z \oplus x (y \oplus z)) == 1$\\	
	&$\frac{2}{3} $&$ a b c (y z \oplus x (y \oplus z)) == 1$\\	
	\hline 
	63
	&$\frac{1}{3} $&$ (1 \oplus c) x y \oplus a (1 \oplus c \oplus b c \oplus z \oplus b z \oplus y z) \oplus b (c \oplus y \oplus x z \oplus y z) == 1$\\	
	&$\frac{2}{3} $&$ a b c x y z == 1$\\	
	\hline 
	64
	&$\frac{1}{3} $&$ x y (c \oplus z) \oplus a (b \oplus c \oplus b c \oplus y z) \oplus b (c \oplus x \oplus x y \oplus z \oplus x z \oplus y z) == 1$\\	
	&$\frac{2}{3} $&$ a b c x y z == 1$\\	
	\hline 
	65
	&$\frac{1}{3} $&$ (a \oplus x) y (c \oplus z) \oplus b (1 \oplus a \oplus c \oplus a c \oplus y \oplus x y z) == 1$\\	
	&$\frac{2}{3} $&$ a b c (1 \oplus y \oplus x y z) == 1$\\	
	\hline 
	66
	&$\frac{1}{3} $&$ x y (c \oplus z) \oplus a (b \oplus c \oplus b c \oplus y z) \oplus b (c \oplus x (1 \oplus y \oplus z)) == 1$\\	
	&$\frac{2}{3} $&$ a b c x y z == 1$\\	
	\hline 
	67
	&$\frac{1}{3} $&$ (a \oplus x) y (c \oplus z) \oplus b (1 \oplus a \oplus c \oplus a c \oplus y \oplus x z) == 1$\\	
	&$\frac{2}{3} $&$ a b c (1 \oplus y \oplus x z) == 1$\\	
	\hline 
	68
	&$\frac{1}{3} $&$ x y (c \oplus z) \oplus b (1 \oplus a \oplus c \oplus a c \oplus y \oplus x z) \oplus a (c \oplus y z) == 1$\\	
	&$\frac{2}{3} $&$ a b c x y z == 1$\\	
	\hline 
	69
	&$\frac{1}{3} $&$ (1 \oplus c) x y \oplus a (1 \oplus c \oplus b c \oplus z \oplus b z \oplus y z) \oplus b (c \oplus x \oplus y \oplus x y \oplus x z \oplus y z) == 1$\\	
	&$\frac{2}{3} $&$ a b c x y z == 1$\\	
	\hline 
	70
	&$\frac{1}{2} $&$ a (b \oplus c \oplus c y) \oplus b x (1 \oplus c \oplus y z) == 1$\\	
	\hline 
	71
	&$\frac{1}{2} $&$ b x (1 \oplus c \oplus y z) \oplus a (b \oplus c \oplus y z) == 1$\\	
	\hline 
	72
	&$\frac{1}{2} $&$ b x (c \oplus y z) \oplus a (c \oplus (1 \oplus b \oplus y) z) == 1$\\	
	\hline 
	73
	&$\frac{1}{2} $&$ c y \oplus a (b \oplus c y) \oplus b (c \oplus x (1 \oplus y) z) == 1$\\	
	\hline 
	74
	&$\frac{1}{2} $&$ a (b \oplus c y) \oplus b (c \oplus y \oplus x z \oplus x y z) == 1$\\	
	\hline 
	75
	&$\frac{1}{2} $&$ a c (1 \oplus y) \oplus b (1 \oplus a \oplus c \oplus x y \oplus y z) == 1$\\	
	\hline 
	76
	&$\frac{1}{2} $&$ x y (c \oplus z) \oplus a (c \oplus y z) \oplus b (1 \oplus a \oplus c \oplus x y \oplus y z) == 1$\\	
	\hline 
	77
	&$\frac{1}{2} $&$ b (c \oplus y) \oplus a (1 \oplus b \oplus c \oplus y) == 1$\\	
	\hline 
	78
	&$\frac{1}{2} $&$ b (c \oplus y) \oplus a (b \oplus c y) == 1$\\	
	\hline 
	79
	&$\frac{1}{2} $&$ a (c x y z \oplus b (1 \oplus y \oplus x y z)) \oplus b (x (1 \oplus y \oplus z) \oplus c (1 \oplus y \oplus x y z)) == 1$\\	
	&$\frac{1}{4} $&$ y (a \oplus b \oplus c \oplus c x \oplus a z \oplus b x z) == 1$\\	
	\hline 
	80
	&$\frac{1}{2} $&$ b (c \oplus c y \oplus x z \oplus c x y z) \oplus a (c x y z \oplus b (1 \oplus y \oplus x y z)) == 1$\\	
	&$\frac{1}{4} $&$ y (a \oplus b \oplus c \oplus c x \oplus a z \oplus b x z) == 1$\\	
	\hline 
	81
	&$\frac{1}{2} $&$ b (c \oplus c y \oplus x y z \oplus c x y z) \oplus a (c x y z \oplus b (1 \oplus y \oplus x y z)) == 1$\\	
	&$\frac{1}{4} $&$ y (a \oplus b \oplus c \oplus c x \oplus a z \oplus b x z) == 1$\\	
	\hline 
	82
	&$\frac{1}{3} $&$ b x (1 \oplus c \oplus y) \oplus a (b (1 \oplus c) \oplus y (c \oplus z)) == 1$\\	
	&$\frac{2}{3} $&$ a b c (1 \oplus y (1 \oplus z \oplus x z)) == 1$\\	
	\hline 
	83
	&$\frac{1}{2} $&$ b (a \oplus c \oplus y (x \oplus z)) == 1$\\	
	\hline 
	84
	&$\frac{1}{2} $&$ b (a \oplus c \oplus y \oplus x y z) == 1$\\	
	\hline 
	85
	&$\frac{1}{2} $&$ b (a \oplus c \oplus x y \oplus x z \oplus y z) == 1$\\	
	\hline 
	86
	&$\frac{1}{2} $&$ b (a \oplus c \oplus (x \oplus y) z) == 1$\\	
	\hline 
	87
	&$\frac{1}{2} $&$ b (a \oplus c \oplus y z) == 1$\\	
	\hline 
	88
	&$\frac{1}{2} $&$ b (a \oplus c \oplus y \oplus x z) == 1$\\	
	\hline 
	89
	&$\frac{1}{2} $&$ b (a \oplus c \oplus y) == 1$\\	
	\hline 
	90
	&$\frac{1}{2} $&$ b (a \oplus c \oplus x y z) == 1$\\	
	\hline 
	91
	&$\frac{2}{3} $&$ a b (1 \oplus c) (1 \oplus y) == 1$\\	
	&$\frac{1}{3} $&$ a y (c \oplus z) \oplus b (c \oplus a c \oplus y \oplus x y z) == 1$\\	
	\hline 
	92
	&$\frac{2}{3} $&$ a b c y == 1$\\	
	&$\frac{1}{3} $&$ c x (1 \oplus y) \oplus b (c \oplus y) \oplus a (b \oplus c \oplus b c \oplus c y) == 1$\\	
	\hline 
	93
	&$\frac{1}{2} $&$ b (a \oplus x (c \oplus y z)) == 1$\\	
	\hline 
	94
	&$\frac{1}{2} $&$ b (a \oplus x (c \oplus y)) == 1$\\	
	\hline 
	95
	&$\frac{1}{2} $&$ b (1 \oplus y) (a \oplus c \oplus x z) == 1$\\	
	&$\frac{1}{4} $&$ y (a \oplus b \oplus c x \oplus b x z) == 1$\\	
	\hline 
	96
	&$\frac{1}{2} $&$ b (1 \oplus y) (a \oplus c \oplus x z) == 1$\\	
	&$\frac{1}{4} $&$ y (a \oplus b \oplus c \oplus a z) == 1$\\	
	\hline 
	97
	&$\frac{1}{2} $&$ b (1 \oplus y) (a \oplus c \oplus x z) == 1$\\	
	&$\frac{1}{4} $&$ y (1 \oplus b \oplus c \oplus a z \oplus b x z) == 1$\\	
	\hline 
	98
	&$\frac{1}{2} $&$ b (1 \oplus y) (a \oplus c \oplus x z) == 1$\\	
	&$\frac{1}{4} $&$ (a \oplus b \oplus c x) y == 1$\\	
	\hline 
	99
	&$\frac{2}{3} $&$ a b c y (x \oplus z \oplus x z) == 1$\\	
	&$\frac{1}{3} $&$ b (c \oplus x \oplus x z \oplus y z) \oplus a (b (1 \oplus c) \oplus (1 \oplus y) (c \oplus z)) == 1$\\	
	\hline 
	100
	&$\frac{2}{3} $&$ a b c y (1 \oplus x z) == 1$\\	
	&$\frac{1}{3} $&$ a (b (1 \oplus c) \oplus (1 \oplus y) (c \oplus z)) \oplus b (c \oplus y \oplus x (1 \oplus y \oplus z)) == 1$\\	
	\hline 
	101
	&$\frac{1}{2} $&$ b (a \oplus c) (1 \oplus y) == 1$\\	
	&$\frac{1}{4} $&$ y (b \oplus c \oplus (a \oplus x) z) == 1$\\	
	\hline 
	102
	&$\frac{1}{2} $&$ b (a \oplus c) (1 \oplus y) == 1$\\	
	&$\frac{1}{4} $&$ y (a \oplus b \oplus x \oplus c x \oplus b x z) == 1$\\	
	\hline 
	103
	&$\frac{2}{3} $&$ a b c y (1 \oplus x z) == 1$\\	
	&$\frac{1}{3} $&$ b (c \oplus y \oplus x z) \oplus a (b (1 \oplus c) \oplus (1 \oplus y) (1 \oplus c \oplus z)) == 1$\\	
	\hline 
	104
	&$\frac{3}{4} $&$ a b c (1 \oplus y) == 1$\\	
	&$\frac{1}{4} $&$ c (1 \oplus x) y \oplus b (1 \oplus c \oplus y \oplus a (1 \oplus c \oplus c y)) \oplus a y z == 1$\\	
	&$\frac{1}{2} $&$ a c y (b z \oplus x (1 \oplus b \oplus z)) == 1$\\	
	\hline 
	105
	&$\frac{2}{3} $&$ a b c y (1 \oplus x z) == 1$\\	
	&$\frac{1}{3} $&$ c x (1 \oplus y) \oplus a (b \oplus c \oplus b c \oplus c y) \oplus b (c \oplus y \oplus x y z) == 1$\\	
	\hline 
	106
	&$\frac{3}{5} $&$ a b (1 \oplus c) (1 \oplus y) == 1$\\	
	&$\frac{2}{5} $&$ a c (1 \oplus x) y z \oplus b ((1 \oplus a \oplus x) y \oplus c (1 \oplus a \oplus x y \oplus a y z \oplus a x y z)) == 1$\\	
	&$\frac{1}{5} $&$ x y (1 \oplus b \oplus c \oplus z) \oplus a y (1 \oplus b \oplus b c \oplus z) == 1$\\	
	\hline 
	107
	&$\frac{2}{3} $&$ a b c y (1 \oplus x z) == 1$\\	
	&$\frac{1}{3} $&$ c x (1 \oplus y) \oplus a (1 \oplus b \oplus b c \oplus y) \oplus b (c \oplus y \oplus x y z) == 1$\\	
	\hline 
	108
	&$\frac{2}{3} $&$ a b c (1 \oplus y) (1 \oplus x z) == 1$\\	
	&$\frac{1}{3} $&$ c (a \oplus x) y \oplus b (1 \oplus a \oplus c \oplus a c \oplus x y \oplus x z \oplus y z) == 1$\\	
	\hline 
	109
	&$\frac{1}{3} $&$ a (1 \oplus c \oplus b c \oplus y z) \oplus b (c \oplus z \oplus x (y \oplus z)) == 1$\\	
	\hline 
	110
	&$\frac{1}{3} $&$ a (1 \oplus c \oplus b c \oplus y z) \oplus b (c \oplus y z \oplus x (y \oplus z)) == 1$\\	
	\hline 
	111
	&$\frac{1}{3} $&$ a (1 \oplus c \oplus b c \oplus y z) \oplus b (c \oplus x \oplus x z \oplus y z) == 1$\\	
	\hline 
	112
	&$\frac{1}{3} $&$ a (b \oplus c \oplus b c) \oplus c x y \oplus b (c \oplus (x \oplus y) z) == 1$\\	
	\hline 
	113
	&$\frac{1}{3} $&$ a (b \oplus c \oplus b c \oplus y \oplus y z) \oplus b (c \oplus (x \oplus y) z) == 1$\\	
	\hline 
	114
	&$\frac{1}{3} $&$ b (c \oplus y \oplus x z) \oplus a (1 \oplus c \oplus b c \oplus y \oplus y z) == 1$\\	
	\hline 
	115
	&$\frac{1}{3} $&$ a (b \oplus c \oplus b c) \oplus c x y \oplus b (c \oplus y z \oplus x (1 \oplus y \oplus z)) == 1$\\	
	\hline 
	116
	&$\frac{1}{3} $&$ b (c \oplus x \oplus x y) \oplus a (b \oplus c \oplus b c \oplus y z) == 1$\\	
	\hline 
	117
	&$\frac{1}{3} $&$ a (1 \oplus c \oplus b c \oplus y z) \oplus b (c \oplus x \oplus y \oplus x y z) == 1$\\	
	\hline 
	118
	&$\frac{1}{3} $&$ a (1 \oplus c \oplus b c \oplus y \oplus y z) \oplus b (c \oplus y \oplus z \oplus x y z) == 1$\\	
	\hline 
	119
	&$\frac{1}{3} $&$ a (1 \oplus c \oplus b c \oplus y \oplus y z) \oplus b (c \oplus y \oplus x z \oplus x y z) == 1$\\	
	\hline 
	120
	&$\frac{1}{3} $&$ a (1 \oplus c \oplus b c \oplus y z) \oplus b (c \oplus y \oplus x z \oplus x y z) == 1$\\	
	\hline 
	121
	&$\frac{1}{3} $&$ a (b \oplus c \oplus b c \oplus z \oplus y z) \oplus b (c \oplus y (x \oplus z)) == 1$\\	
	\hline 
	122
	&$\frac{1}{3} $&$ a (b \oplus c \oplus b c) \oplus c x y \oplus b (c \oplus x \oplus y \oplus z \oplus x y z) == 1$\\	
	\hline 
	123
	&$\frac{1}{2} $&$ b (1 \oplus c) (a \oplus x y) == 1$\\	
	&$\frac{1}{4} $&$ c x (1 \oplus y) \oplus b (c \oplus c x \oplus z \oplus y z) \oplus a (c \oplus (b \oplus y) z) == 1$\\	
	\hline 
	124
	&$\frac{1}{2} $&$ b (1 \oplus a \oplus x) (c \oplus y z) == 1$\\	
	&$\frac{1}{4} $&$ b x (c \oplus y \oplus z) \oplus a (1 \oplus b \oplus c \oplus y \oplus b z) == 1$\\	
	\hline 
	125
	&$\frac{1}{3} $&$ a (b \oplus c \oplus b c) \oplus c x (1 \oplus y) \oplus b (c \oplus x y z) == 1$\\	
	\hline 
	126
	&$\frac{1}{3} $&$ a (b \oplus c \oplus b c \oplus y \oplus y z) \oplus b (c \oplus x (1 \oplus y) z) == 1$\\	
	\hline 
	127
	&$\frac{1}{3} $&$ a (b \oplus c \oplus b c) \oplus c x y \oplus b (c \oplus (1 \oplus x) y z) == 1$\\	
	\hline 
	128
	&$\frac{1}{2} $&$ b c (1 \oplus a \oplus x y) == 1$\\	
	&$\frac{1}{4} $&$ x y (b \oplus c \oplus z) \oplus a (b \oplus c \oplus y z) == 1$\\	
	\hline 
	129
	&$\frac{1}{2} $&$ b (1 \oplus c \oplus y \oplus c y z \oplus a (1 \oplus x y \oplus c (1 \oplus y (1 \oplus x \oplus z)))) == 1$\\	
	&$\frac{1}{4} $&$ b (y (x \oplus z) \oplus c (x \oplus x y z)) \oplus a ((1 \oplus b \oplus y \oplus b x y) z \oplus c (1 \oplus y \oplus b x y z)) == 1$\\	
	&$\frac{3}{4} $&$ (1 \oplus a) b (1 \oplus c) x y z == 1$\\	
	\hline 
	130
	&$\frac{1}{2} $&$ b ((1 \oplus a) (1 \oplus x) y z \oplus c (1 \oplus a \oplus x \oplus a x y \oplus a y z)) == 1$\\	
	&$\frac{1}{4} $&$ b x (c \oplus y \oplus z \oplus y z) \oplus a ((1 \oplus c) (1 \oplus y) \oplus b (1 \oplus c x y) (1 \oplus z)) == 1$\\	
	&$\frac{3}{4} $&$ a b c x y (1 \oplus z) == 1$\\	
	\hline 
	131
	&$\frac{1}{2} $&$ b ((1 \oplus a \oplus x) y z \oplus c (1 \oplus a \oplus x \oplus a x y \oplus a y z)) == 1$\\	
	&$\frac{1}{4} $&$ b x (c \oplus y \oplus z) \oplus a ((1 \oplus c) (1 \oplus y) \oplus b (1 \oplus c x y) (1 \oplus z)) == 1$\\	
	&$\frac{3}{4} $&$ a b c x y (1 \oplus z) == 1$\\	
	\hline 
	132
	&$\frac{1}{2} $&$ b c (1 \oplus a \oplus x \oplus y \oplus a y \oplus x y \oplus a x y z) == 1$\\	
	&$\frac{1}{4} $&$ a (1 \oplus c \oplus c y \oplus b z) \oplus y (c \oplus z \oplus x z) \oplus b (y \oplus x (1 \oplus c \oplus z \oplus y z)) == 1$\\	
	\hline 
	133
	&$\frac{1}{2} $&$ (1 \oplus a) b c (1 \oplus y \oplus x y z) == 1$\\	
	&$\frac{1}{4} $&$ y (c \oplus x \oplus x z) \oplus a (1 \oplus c \oplus c y \oplus z \oplus b z \oplus y z) \oplus b (x \oplus c x \oplus y \oplus y z \oplus x y z) == 1$\\	
	\hline 
	134
	&$\frac{1}{2} $&$ a b ((x \oplus y \oplus x y) z \oplus c (y \oplus x z \oplus x y z)) == 1$\\	
	&$\frac{1}{4} $&$ a \oplus c \oplus b c x \oplus y \oplus a c y \oplus b x y \oplus b z \oplus a b z \oplus x z \oplus x y z \oplus b x y z == 1$\\	
	\hline 
	135
	&$\frac{1}{2} $&$ a b c (y \oplus x z \oplus x y z) == 1$\\	
	&$\frac{1}{4} $&$ a \oplus c \oplus b c x \oplus y \oplus a c y \oplus b x y \oplus a b z \oplus x z \oplus b x z \oplus y z \oplus a y z \oplus b y z \oplus x y z \oplus b x y z == 1$\\	
	\hline 
	136
	&$\frac{1}{3} $&$ c x y \oplus a (b \oplus b c \oplus c y) \oplus b (c \oplus x \oplus y \oplus y z \oplus x y z) == 1$\\	
	&$\frac{2}{3} $&$ a b c x (1 \oplus y) == 1$\\	
	\hline 
	137
	&$\frac{1}{3} $&$ a (1 \oplus c \oplus b c \oplus y \oplus b z) \oplus b (c \oplus y \oplus x y z) == 1$\\	
	\hline 
	138
	&$\frac{1}{3} $&$ (1 \oplus c) x y \oplus a (1 \oplus c (1 \oplus b \oplus y) \oplus b z) \oplus b (c \oplus y \oplus y z \oplus x y z) == 1$\\	
	\hline 
	139
	&$\frac{1}{3} $&$ b (c \oplus x \oplus y \oplus y z \oplus x y z) \oplus a (b (1 \oplus c) \oplus y (c \oplus z)) == 1$\\	
	&$\frac{2}{3} $&$ a b c x (1 \oplus y) == 1$\\	
	\hline 
	140
	&$\frac{1}{4} $&$ a (c \oplus c y \oplus b z) \oplus b (c \oplus y \oplus z \oplus x y z) == 1$\\	
	&$\frac{1}{2} $&$ a (y (c \oplus z) \oplus b (1 \oplus c \oplus y \oplus y z)) == 1$\\	
	\hline 
	141
	&$\frac{1}{3} $&$ b (c \oplus x y) \oplus a (b (1 \oplus c) \oplus (1 \oplus y) (c \oplus z)) == 1$\\	
	&$\frac{2}{3} $&$ a b c x y == 1$\\	
	\hline 
	142
	&$\frac{1}{4} $&$ b (c \oplus x \oplus c x \oplus y \oplus z \oplus x z) \oplus a (b \oplus y (c \oplus z)) == 1$\\	
	&$\frac{1}{2} $&$ b c x y \oplus a c (1 \oplus b \oplus y) == 1$\\	
	\hline 
	143
	&$\frac{1}{3} $&$ b (c \oplus x \oplus y \oplus z \oplus x y z) \oplus a (b (1 \oplus c) \oplus y (c \oplus z)) == 1$\\	
	&$\frac{2}{3} $&$ a b c (1 \oplus y) (x \oplus z) == 1$\\	
	\hline 
	144
	&$\frac{1}{4} $&$ b (a \oplus c \oplus x \oplus c x \oplus y \oplus z \oplus x z) == 1$\\	
	&$\frac{1}{2} $&$ (1 \oplus b) c (a \oplus x y) == 1$\\	
	\hline 
	145
	&$\frac{1}{4} $&$ b (a \oplus c \oplus x \oplus c x \oplus y \oplus y z \oplus x y z) == 1$\\	
	&$\frac{1}{2} $&$ (1 \oplus b) c (a \oplus x y) == 1$\\	
	\hline 
	146
	&$\frac{1}{3} $&$ c (1 \oplus x) y \oplus b (c \oplus x y) \oplus a (b \oplus c \oplus b c \oplus z \oplus y z) == 1$\\	
	\hline 
	147
	&$\frac{1}{2} $&$ b (1 \oplus c) (x y \oplus a (1 \oplus (1 \oplus x) y z)) == 1$\\	
	&$\frac{1}{4} $&$ c x (1 \oplus y) \oplus b (c \oplus c x \oplus z \oplus x y z) \oplus a (c \oplus (b \oplus y) z) == 1$\\	
	\hline 
	148
	&$\frac{1}{2} $&$ b c (a \oplus x \oplus x y \oplus a y z \oplus a x y z) == 1$\\	
	&$\frac{1}{4} $&$ x y (c \oplus z) \oplus a (1 \oplus b \oplus c \oplus y \oplus b z \oplus y z) \oplus b (c \oplus c x \oplus y \oplus z \oplus x z \oplus x y z) == 1$\\	
	\hline 
	149
	&$\frac{1}{2} $&$ b c (a \oplus x \oplus x y \oplus a x y z) == 1$\\	
	&$\frac{1}{4} $&$ x y (c \oplus z) \oplus a (1 \oplus b \oplus c \oplus y \oplus b z \oplus y z) \oplus b (c \oplus c x \oplus y \oplus z \oplus x z \oplus y z \oplus x y z) == 1$\\	
	\hline 
	150
	&$\frac{1}{2} $&$ b (1 \oplus c) (a \oplus a x y z \oplus x y (1 \oplus z)) == 1$\\	
	&$\frac{1}{4} $&$ c x (1 \oplus y) \oplus b (c \oplus c x \oplus z \oplus y z \oplus x y z) \oplus a (c \oplus (b \oplus y) z) == 1$\\	
	\hline 
	151
	&$\frac{1}{2} $&$ b c (1 \oplus a \oplus y \oplus x y \oplus a y z \oplus a x y z) == 1$\\	
	&$\frac{1}{4} $&$ y (1 \oplus c \oplus z) \oplus a (1 \oplus c \oplus z \oplus b z) \oplus b x (1 \oplus c \oplus y z) == 1$\\	
	\hline 
	152
	&$\frac{1}{2} $&$ b c (1 \oplus a \oplus x \oplus y \oplus a y \oplus x y \oplus a x y z) == 1$\\	
	&$\frac{1}{4} $&$ a (1 \oplus c \oplus c y \oplus b z) \oplus y (1 \oplus c \oplus z \oplus x z) \oplus b x (1 \oplus c \oplus z \oplus y z) == 1$\\	
	\hline 
	153
	&$\frac{1}{2} $&$ b ((1 \oplus a \oplus x) y z \oplus c (1 \oplus a \oplus x \oplus a x y \oplus a x y z)) == 1$\\	
	&$\frac{1}{4} $&$ a (1 \oplus b \oplus c \oplus y \oplus b z) \oplus b x (c \oplus z \oplus y z) == 1$\\	
	\hline 
	154
	&$\frac{1}{2} $&$ b c ((1 \oplus x) (1 \oplus y) \oplus a (1 \oplus y \oplus x y \oplus x y z)) == 1$\\	
	&$\frac{1}{4} $&$ c y \oplus x y z \oplus a (1 \oplus b \oplus c \oplus c y \oplus b z) \oplus b (c x \oplus (x \oplus y \oplus x y) z) == 1$\\	
	\hline 
	155
	&$\frac{1}{2} $&$ a b (1 \oplus c \oplus y \oplus c x y \oplus y z \oplus c x y z) == 1$\\	
	&$\frac{1}{4} $&$ x (c \oplus y) \oplus a (c \oplus y \oplus b z) \oplus b (c \oplus c x \oplus y \oplus z \oplus x y z) == 1$\\	
	\hline 
	156
	&$\frac{1}{2} $&$ b c (1 \oplus a \oplus x \oplus y \oplus a y \oplus x y \oplus a x y z) == 1$\\	
	&$\frac{1}{4} $&$ c y \oplus x y z \oplus a (1 \oplus c \oplus c y \oplus b z) \oplus b (y (1 \oplus z) \oplus x (1 \oplus c \oplus z \oplus y z)) == 1$\\	
	\hline 
	157
	&$\frac{1}{2} $&$ a c (b x z \oplus y (1 \oplus b \oplus b x z)) == 1$\\	
	&$\frac{1}{4} $&$ a \oplus c \oplus x \oplus b x \oplus b c x \oplus y \oplus a c y \oplus c x y \oplus a b z \oplus x z \oplus b x z \oplus b y z \oplus x y z \oplus b x y z == 1$\\	
	\hline 
	158
	&$\frac{1}{3} $&$ c x y \oplus a (b \oplus c \oplus b c \oplus z \oplus y z) \oplus b (c \oplus y z \oplus x (1 \oplus y \oplus z)) == 1$\\	
	\hline 
	159
	&$\frac{1}{3} $&$ c (a \oplus y) \oplus b (1 \oplus a \oplus a c \oplus c x \oplus x y \oplus x z \oplus y z) == 1$\\	
	\hline 
	160
	&$\frac{1}{2} $&$ b c (1 \oplus a \oplus y \oplus a x y \oplus x y z \oplus a x y z) == 1$\\	
	&$\frac{1}{4} $&$ y (b \oplus c \oplus c x \oplus b x z) \oplus a (b \oplus c \oplus z \oplus y z) == 1$\\	
	\hline 
	161
	&$\frac{1}{2} $&$ b c (1 \oplus a \oplus y \oplus x y \oplus x y z \oplus a x y z) == 1$\\	
	&$\frac{1}{4} $&$ y (b \oplus c \oplus b x \oplus c x \oplus b x z) \oplus a (b \oplus c \oplus z \oplus y z) == 1$\\	
	\hline 
	162
	&$\frac{1}{2} $&$ b c (1 \oplus a \oplus x y \oplus a x y \oplus a x y z) == 1$\\	
	&$\frac{1}{4} $&$ x y (c \oplus b z) \oplus a (b \oplus c \oplus z \oplus y z) == 1$\\	
	\hline 
	163
	&$\frac{1}{2} $&$ b c (1 \oplus a \oplus x y \oplus a x y z) == 1$\\	
	&$\frac{1}{4} $&$ x y (b \oplus c \oplus b z) \oplus a (b \oplus c \oplus z \oplus y z) == 1$\\	
	\hline 
	164
	&$\frac{1}{2} $&$ a b (1 \oplus c \oplus y \oplus y z \oplus x y z \oplus c x y z) == 1$\\	
	&$\frac{1}{4} $&$ x (c \oplus y) \oplus a (c \oplus y \oplus b z) \oplus b (c \oplus c x \oplus y \oplus x y \oplus z \oplus x y z) == 1$\\	
	\hline 
	165
	&$\frac{1}{2} $&$ a b (1 \oplus c \oplus c x y \oplus z \oplus c y z) \oplus b x (1 \oplus c \oplus y \oplus z \oplus y z \oplus c y z) == 1$\\	
	&$\frac{1}{4} $&$ b (c \oplus c x \oplus y z) \oplus a ((1 \oplus b \oplus y) z \oplus c (1 \oplus y \oplus b y z \oplus b x y z)) == 1$\\	
	&$\frac{3}{4} $&$ a b c (1 \oplus x) y z == 1$\\	
	\hline 
	166
	&$\frac{3}{5} $&$ b c (a \oplus x) (1 \oplus y) == 1$\\	
	&$\frac{1}{5} $&$ a (1 \oplus b \oplus c \oplus y \oplus c y \oplus b c y \oplus b z) \oplus b (c \oplus c x \oplus y \oplus c x y \oplus z \oplus x z \oplus x y z) == 1$\\	
	&$\frac{2}{5} $&$ b x y (c \oplus z) \oplus a y (c \oplus b z) == 1$\\	
	\hline 
	167
	&$\frac{3}{5} $&$ a b c (1 \oplus y) == 1$\\	
	&$\frac{1}{5} $&$ (1 \oplus c) x (1 \oplus y) \oplus a (1 \oplus b \oplus c \oplus y \oplus c y \oplus b c y \oplus b z) \oplus b (c \oplus x \oplus c x \oplus y \oplus z \oplus x y z) == 1$\\	
	&$\frac{2}{5} $&$ b c x y \oplus a y (c \oplus z \oplus b z) == 1$\\	
	\hline 
	168
	&$\frac{2}{5} $&$ a y (1 \oplus c \oplus z) \oplus b ((1 \oplus y) (1 \oplus x z) \oplus c (1 \oplus y z \oplus x y (1 \oplus z)) \oplus a (1 \oplus y z \oplus c (1 \oplus x z \oplus y (1 \oplus x \oplus z)))) == 1$\\	
	&$\frac{1}{5} $&$ y (c x \oplus (b \oplus x) z) \oplus a (y z \oplus c (y \oplus b y \oplus b x z \oplus b x y z)) == 1$\\	
	&$\frac{3}{5} $&$ a b c (1 \oplus y) (1 \oplus x z) == 1$\\	
	\hline 
	169
	&$\frac{2}{5} $&$ (1 \oplus a) c x y (1 \oplus z) \oplus b ((1 \oplus y) (1 \oplus x z) \oplus a (1 \oplus c \oplus c x y \oplus c x z \oplus y z) \oplus c (1 \oplus y \oplus x y \oplus x y z)) == 1$\\	
	&$\frac{1}{5} $&$ (a \oplus b \oplus x) y z \oplus c (a b x z \oplus y (1 \oplus a b \oplus x \oplus a b x z)) == 1$\\	
	&$\frac{3}{5} $&$ a b c (1 \oplus y) (1 \oplus x z) == 1$\\	
	\hline 
	170
	&$\frac{2}{5} $&$ c (a \oplus x) y \oplus b (1 \oplus c \oplus y \oplus c y \oplus c x y \oplus x z \oplus x y z \oplus a (1 \oplus c \oplus y \oplus x y \oplus c x y \oplus c x z \oplus x y z)) == 1$\\	
	&$\frac{1}{5} $&$ (1 \oplus a \oplus x) y (c \oplus z) \oplus b (a c x z \oplus y (1 \oplus a c \oplus x \oplus a c x z)) == 1$\\	
	&$\frac{3}{5} $&$ a b c (1 \oplus y) (1 \oplus x z) == 1$\\	
	\hline 
	171
	&$\frac{2}{5} $&$ a (1 \oplus c) x y z \oplus b (1 \oplus c \oplus y \oplus c y \oplus c x y \oplus x z \oplus x y z \oplus a (1 \oplus c \oplus y \oplus c x z \oplus x y z)) == 1$\\	
	&$\frac{1}{5} $&$ (1 \oplus x) y (b \oplus c \oplus z) \oplus a (b c x z \oplus y (1 \oplus b c \oplus z \oplus b c x z)) == 1$\\	
	&$\frac{3}{5} $&$ a b c (1 \oplus y) (1 \oplus x z) == 1$\\	
	\hline 
	172
	&$\frac{2}{5} $&$ a (b \oplus y \oplus c y \oplus b c y \oplus b c x y z) \oplus b (x (1 \oplus y) z \oplus c (1 \oplus y \oplus x y z)) == 1$\\	
	&$\frac{1}{5} $&$ y (b \oplus c \oplus c x \oplus b z) \oplus a (1 \oplus c \oplus y \oplus z \oplus y z \oplus b (1 \oplus c \oplus z)) == 1$\\	
	\hline 
	173
	&$\frac{2}{5} $&$ a (1 \oplus c) y \oplus b ((1 \oplus y) (1 \oplus x z) \oplus c (1 \oplus y z \oplus x y (1 \oplus z)) \oplus a (1 \oplus c y (x \oplus z \oplus x z))) == 1$\\	
	&$\frac{1}{5} $&$ c x y \oplus b y z \oplus a ((1 \oplus b) c \oplus (1 \oplus b \oplus y) z) == 1$\\	
	\hline 
	174
	&$\frac{2}{5} $&$ (1 \oplus b) (c \oplus y \oplus x z \oplus x y z) \oplus a (1 \oplus c y \oplus b (1 \oplus c y (1 \oplus z \oplus x z))) == 1$\\	
	&$\frac{1}{5} $&$ b x (1 \oplus c \oplus y) \oplus a (b (1 \oplus c) \oplus y (c \oplus z)) == 1$\\	
	\hline 
	175
	&$\frac{2}{5} $&$ b (c \oplus c x y \oplus (x \oplus y) z) \oplus a (c (1 \oplus x) y z \oplus b (1 \oplus y (1 \oplus c \oplus z \oplus c z \oplus c x z))) == 1$\\	
	&$\frac{1}{5} $&$ (b \oplus c) x y \oplus a (y \oplus c (1 \oplus b \oplus y) \oplus z \oplus b z) == 1$\\	
	\hline 
	176
	&$\frac{2}{5} $&$ b (x (1 \oplus y) (1 \oplus z) \oplus c (1 \oplus y z \oplus x y (1 \oplus z))) \oplus a ((1 \oplus c) y \oplus b (1 \oplus c y (x \oplus z \oplus x z))) == 1$\\	
	&$\frac{1}{5} $&$ c x y \oplus b y z \oplus a ((1 \oplus b) c \oplus (1 \oplus b \oplus y) z) == 1$\\	
	\hline 
	177
	&$\frac{2}{5} $&$ b (x (1 \oplus y) z \oplus c (1 \oplus y z \oplus x y (1 \oplus z))) \oplus a (y (1 \oplus c \oplus z) \oplus b (1 \oplus y z \oplus c y (x \oplus z \oplus x z))) == 1$\\	
	&$\frac{1}{5} $&$ y (c x \oplus (b \oplus x) z) \oplus a (c \oplus b c \oplus y z) == 1$\\	
	\hline 
	178
	&$\frac{2}{5} $&$ b (x (1 \oplus y) (1 \oplus z) \oplus c (1 \oplus y z \oplus x y (1 \oplus z))) \oplus a (y (1 \oplus c \oplus z) \oplus b (1 \oplus y z \oplus c y (x \oplus z \oplus x z))) == 1$\\	
	&$\frac{1}{5} $&$ y (c x \oplus (b \oplus x) z) \oplus a (c \oplus b c \oplus y z) == 1$\\	
	\hline 
	179
	&$\frac{2}{5} $&$ c y \oplus b (c \oplus x (1 \oplus y) z) \oplus a (c y \oplus b (1 \oplus (1 \oplus c) y z \oplus (1 \oplus c) x y (1 \oplus z))) == 1$\\	
	&$\frac{1}{5} $&$ b x y \oplus c x (1 \oplus b \oplus y) \oplus a (c \oplus b c \oplus y z) == 1$\\	
	\hline 
	180
	&$\frac{1}{2} $&$ a (c (1 \oplus b \oplus y) \oplus b y z) == 1$\\	
	&$\frac{1}{4} $&$ c x y \oplus a (b \oplus c y \oplus b z) \oplus b (c \oplus y \oplus x y \oplus z \oplus x y z) == 1$\\	
	\hline 
	181
	&$\frac{3}{5} $&$ a b c y == 1$\\	
	&$\frac{1}{5} $&$ (1 \oplus c) (a \oplus x) y \oplus b (y \oplus c (1 \oplus a \oplus a y)) == 1$\\	
	&$\frac{2}{5} $&$ a (1 \oplus y) (b \oplus c \oplus z) == 1$\\	
	\hline 
	182
	&$\frac{1}{2} $&$ b (a \oplus x) (c \oplus y z) == 1$\\	
	&$\frac{1}{4} $&$ c y \oplus b (1 \oplus a \oplus c \oplus y \oplus a z) \oplus a (c \oplus z \oplus y z) == 1$\\	
	\hline 
	183
	&$\frac{1}{2} $&$ a b (c \oplus y z) == 1$\\	
	&$\frac{1}{4} $&$ c x y \oplus a (1 \oplus b \oplus c \oplus y \oplus b z) \oplus b (c \oplus y \oplus x y \oplus z \oplus x y z) == 1$\\	
	\hline 
	184
	&$\frac{1}{3} $&$ a (b (1 \oplus c) \oplus (1 \oplus y) (c \oplus z)) \oplus b (c \oplus y (x \oplus z)) == 1$\\	
	&$\frac{2}{3} $&$ a b c y (x \oplus z) == 1$\\	
	\hline 
	185
	&$\frac{1}{4} $&$ y (a \oplus c \oplus c x \oplus a z) \oplus b (1 \oplus c \oplus y \oplus x z \oplus x y z \oplus a (1 \oplus c (1 \oplus y) (1 \oplus x z))) == 1$\\	
	&$\frac{1}{2} $&$ a c (x y z \oplus b (x z \oplus (1 \oplus x) y (1 \oplus z))) == 1$\\	
	&$\frac{3}{4} $&$ a b c (1 \oplus y) (1 \oplus x z) == 1$\\	
	\hline 
	186
	&$\frac{1}{4} $&$ c x y \oplus a y (1 \oplus z) \oplus b (1 \oplus c \oplus y \oplus x z \oplus x y z \oplus a (1 \oplus c (1 \oplus y) (1 \oplus x z))) == 1$\\	
	&$\frac{1}{2} $&$ a c ((1 \oplus x) y z \oplus b (y z \oplus x (y \oplus z \oplus y z))) == 1$\\	
	&$\frac{3}{4} $&$ a b c (1 \oplus y) (1 \oplus x z) == 1$\\	
	\hline 
	187
	&$\frac{1}{4} $&$ c x y \oplus a y z \oplus b (1 \oplus a \oplus c \oplus y \oplus x z \oplus x y z) == 1$\\	
	&$\frac{1}{2} $&$ a c (1 \oplus b \oplus x y \oplus b x y \oplus y z \oplus b y z \oplus x y z) == 1$\\	
	\hline 
	188
	&$\frac{1}{4} $&$ c x y \oplus a (b \oplus y \oplus y z) \oplus b (c \oplus x (1 \oplus y) z) == 1$\\	
	&$\frac{1}{2} $&$ a c (1 \oplus y \oplus y z \oplus x y z \oplus b (1 \oplus y (1 \oplus x \oplus z))) == 1$\\	
	\hline 
	189
	&$\frac{1}{4} $&$ c x y \oplus a y z \oplus b (a \oplus c \oplus (1 \oplus y) (1 \oplus z \oplus x z)) == 1$\\	
	&$\frac{1}{2} $&$ a c (1 \oplus b \oplus x y \oplus b x y \oplus y z \oplus b y z \oplus x y z) == 1$\\	
	\hline 
	190
	&$\frac{1}{4} $&$ c x y \oplus a (b \oplus y z) \oplus b (c \oplus x (1 \oplus y) z) == 1$\\	
	&$\frac{1}{2} $&$ a c (1 \oplus b \oplus x y \oplus b x y \oplus y z \oplus b y z \oplus x y z) == 1$\\	
	\hline 
\end{longtable}

 \renewcommand*{\arraystretch}{1.4}
 \begin{longtable}{ccl}
 \hline
  Class & Prob. & Condition for  RC Extremal Boxes which are also in the No-signaling  Polytope \\ \\
  \hline
  \endfirsthead
  \multicolumn{3}{r}%
  {\tablename\ \thetable\ -- \textit{Continuation form previous page...}} \\
  \hline
  Class & Prob. & Condition  \\
  \hline
  \endhead
  \hline \multicolumn{3}{r}{\textit{Continuation on next page...}} \\
  \endfoot
  \hline
  \endlastfoot
  
1
&$1 $&$ a b c == 1$\\	

 \hline 
2
&$\frac{1}{3} $&$  (1 \oplus x) y (c \oplus z) \oplus b (c \oplus z \oplus x z) \oplus a (b \oplus c \oplus b c \oplus y z)  == 1$\\	
&$\frac{2}{3} $&$ a b c (1 \oplus x) y z  == 1$\\	
 \hline 
3
&$\frac{1}{2} $&$ a (b \oplus c \oplus y z) == 1$\\	
	
 \hline 
4
&$\frac{1}{2} $&$ b (a \oplus c \oplus x z)  == 1$\\	

 \hline 
5
&$\frac{1}{3} $&$  a (b \oplus c \oplus b c) \oplus c x y \oplus b (c \oplus z \oplus x z) == 1$\\	

 \hline 
6
&$\frac{1}{3} $&$ c (b \oplus y \oplus x y) \oplus a (b \oplus c \oplus b c \oplus z \oplus y z) == 1$\\	

 \hline 
 \end{longtable}

	\bibliographystyle{unsrt}
	\bibliography{References-2}

\end{document}